\documentclass[a4paper,normalheadings,pointlessnumbers,11pt]{scrartcl}
\pdfoutput=1
\usepackage{graphicx, color}

\usepackage[paper=a4paper,left=30mm,right=30mm,top=30mm,bottom=30mm]{geometry}

\usepackage[square,numbers]{natbib}
\usepackage{lmodern}
\usepackage[utf8]{inputenc}
\usepackage[T1]{fontenc}
\usepackage[english]{babel}
\usepackage{setspace}
\usepackage{graphicx}
\usepackage{caption}
\usepackage{hyperref}
\usepackage{color}
\usepackage{subcaption}
\usepackage{graphicx}
\usepackage{alltt}
\usepackage{enumerate}
\usepackage{booktabs}
\usepackage{amsmath}
\usepackage{amsfonts}
\usepackage{amssymb}
\usepackage{amsthm}
\usepackage{hyperref}
\usepackage{caption}
\usepackage{paralist}
\usepackage{subcaption}
\usepackage{scrpage2}

\usepackage{amsmath}
\usepackage{amsfonts}
\usepackage{amssymb}
\usepackage{amsthm}

\theoremstyle{plain}
\newtheorem{theorem}{Theorem}

\newtheorem{corollary}[theorem]{Corollary}

\newtheorem{assumption}[theorem]{Assumption}
\theoremstyle{definition}   
\newtheorem{definition}[theorem]{Definition}
\newtheorem{example}[theorem]{Example}

\newcommand{\TP}{\textbf{TP}}
\newcommand{\FP}{\textbf{FP}}

\newcommand{\FN}{\textbf{FN}}
\newcommand{\pre}{\text{prec}}
\newcommand{\rec}{\text{rec}}

\title{A graph theoretic linkage attack on\\microdata in a metric space}
\author{Martin Kroll\\ \href{mailto:martin.kroll@uni-due.de}{\texttt{martin.kroll@uni-due.de}}}

\date{February 13, 2014}

\begin{document}

\definecolor{DodgerBlue4}{RGB}{016,078,139}

\addtokomafont{title}{\color{DodgerBlue4}}
\addtokomafont{section}{\color{DodgerBlue4}}

\pagestyle{scrheadings}

\maketitle

\begin{abstract}
Certain methods of analysis require the knowledge of the spatial distances between entities whose data are stored in a microdata table. For instance, such knowledge is necessary and sufficient to perform data mining tasks such as nearest neighbour searches or clustering. However, when inter-record distances are published in addition to the microdata for research purposes, the risk of identity disclosure has to be taken into consideration again. In order to tackle this problem, we introduce a flexible graph model for microdata in a metric space and propose a linkage attack based on realistic assumptions of a data snooper's background knowledge. This attack is based on the idea of finding a maximum approximate common subgraph of two vertex-labelled and edge-weighted graphs. By adapting a standard argument from algorithmic graph theory to our setup, this task is transformed to the maximum clique detection problem in a corresponding product graph. Using a toy example and experimental results on simulated data show that publishing even approximate distances could increase the risk of identity disclosure unreasonably.
\end{abstract}


\section{Introduction}\label{sec:introduction}

Enriching microdata with spatial information opens up numerous additional approaches for analysis. In the area of epidemiology, this insight goes back at least to the middle of the 19th century when John Snow identified a contaminated water pump in London as the source of a cholera outbreak by linking the cases of mortality to their location and visualising these locations and the positions of surrounding water pumps on a map~\cite{snow}.

\smallskip

In recent years, spatial analysis techniques have become increasingly attractive in the social sciences as well~\cite{parker}. However, when personal microdata containing sensitive information (e.g., gathered in a survey or health study) are published for research purposes, the anonymity of the individuals has to be guaranteed.
It has been pointed out in~\cite{elemam} that
\textit{location is often one of the critical pieces of information for a successful re-identification attack}.
Therefore, usually only microdata that contain spatial information in an aggregated form are released, which restricts the choice of applicable techniques for analysis drastically.

\smallskip

In particular, distance calculations that are based on aggregated data become difficult and imprecise~\cite{beyer}, especially for entities that are closely related. Since many data mining techniques and methods in spatial analysis require accurate distance computations, it is necessary to investigate the extent to which additionally published (approximate) inter-record distances influence the risk of identity disclosure and how a possible non-acceptable increase of this risk can be prevented. Our work presented in this article provides a novel approach for tackling these questions.

\subsection*{Contributions of the paper}

We introduce a flexible natural graph model for microdata with known inter-record distances. The search for a maximum common subgraph between two such graph models is interpreted as a novel kind of linkage attack on such microdata. We discuss the relative merits of our method in comparison to the usual linkage attacks on the basis of a small-scale example (example~\ref{example:poets} in section~\ref{sec:linkage_attack}).

\smallskip

Furthermore, in the special case of geographical distances, it is shown that, on the basis of simulated data, a non-negligible risk of identity disclosure exists if $\mathcal N(0,\sigma^2)$-distributed Gaussian noise is added to the input coordinates for too small values of $\sigma$. For larger values of $\sigma$ (which lead to sufficiently anonymised data), however, the data become nearly useless for further analysis. These results reflect a trade-off between data utility and disclosure risk through the proposed attack.

\subsection*{Organisation of the paper} In section \ref{sec:background}, we refer to related work. The preparatory work is given in section \ref{sec:graph_model} as well as a graph model for microdata in a metric space which forms the basis of the graph theoretic linkage attack introduced in section \ref{sec:linkage_attack}. In section \ref{sec:experimental_results}, this attack is evaluated by means of a simulation study. We conclude and discuss the possible directions for future research in section \ref{sec:conclusion}.

\section{Related work}\label{sec:background}

\subsection*{Statistical disclosure control and privacy preserving data mining}

As already indicated in the introductory section above, the original motivation for the work presented in this article stems back to the wish to also make the wide variety of distance-based methods (e.g. from spatial statistics) applicable for microdata that are published for scientific purposes.
Since it is intuitively compelling that naive release of the exact distances between individuals can increase the risk of deanonymisation, the interest question, however, is how might the knowledge of approximate distances change the risk of identity disclosure, i.e. the chance of a data snooper attempting to identify some of the entities.

\smallskip

In general, the analysis of such deanonymisation attacks on microdata and the development of tools for their anonymisation is a central topic of \textit{statistical disclosure control}~\cite{hundepool}. It is universally acknowledged that a necessary but insufficient first step during the proc\-ess of anonymisation consists in the removal of all attributes that can be used to identify an individual entity unambiguously (this step is usually referred to as \textit{deidentification}). Such attributes (e.g., \texttt{social insurance number}) are called \textit{(direct) identifiers}, in contrast to \textit{quasi-identifiers}, which do not have the power to nullify an individual's anonymity on their own, a distinction which has to be ascribed to Dalenius~\cite{dalenius}.

\smallskip

By using a combination of quasi-identifiers, however, it might be possible to assign an entity from the underlying population to a specific record of a published microdata file unambigously. For example, in~\cite{sweeney2000} it was shown that based on 1990 US census data, 87\% of the population of the United States are uniquely determined by their values with respect to the quasi-identifier set \{\texttt{5-digit ZIP code}, \texttt{gender}, \texttt{date of birth}\}. This fact motivates a mode of attack that is commonly referred to as \textit{linkage attack}~\cite{duncan}: In this scenario, it is assumed that a data snooper has access to an external auxiliary microdata file (called \textit{identification file}) containing both direct identifiers and quasi-identifiers as attributes. By making use of the quasi-identifiers, the snooper attempts to identify entities by linking records from the identification file to records from the published microdata file (termed \textit{target file}). A real-life example of linkage via quasi-identifiers is due to Sweeney~\cite{sweeney2002}: She was able to detect the record corresponding to the governor of Massachusetts in a published health data file by linkage with a publicly obtainable voter registration list.

\smallskip

Even though theoretical results on linkage attacks were recently obtained in~\cite{merener}, the concept of $k$-anonymity had already been proposed as a remedy against linkage attacks in~\cite{samarati}. The basic idea of $k$-anonymity is to modify the records in the released microdata such that every record coincides with at least $k-1$ other records with respect to the quasi-identifiers. For this reason, an unambiguous linkage between the identification and target file will not be possible. The graph theoretic linkage attack introduced in section \ref{sec:linkage_attack} contains the classical linkage attack via quasi-identifiers as a subroutine, however, it provides a way to resolve at least some of the ambiguous matches.

\smallskip

Several papers on \textit{privacy preserving data mining} have already discussed privacy issues with respect to the distance-preserving transformations of microdata. However, in these articles it is generally assumed that the considered distances can be directly calculated from the microdata, whereas our focus is on microdata enriched with supplementary distances between the entities that cannot be calculated from the microdata itself. Moreover, in most cases only specific kinds of distances have been considered (e.g., $\ell_1$-distance in~\cite{rane} or the Euclidean (i.e. $\ell_2$-) distance in~\cite{liub}).

\smallskip

In contrast, the attack proposed in this paper can be applied to any kind of distance function (notwithstanding that the special case of spatial distances motivated our research and is exclusively referred to in our examples). Furthermore, a distance-preserving technique for the anonymisation of binary vectors is discussed in~\cite{kenthapadi}. In contrast to our approach, in that article the distance information alone is not assumed to increase the risk of identity disclosure.

\subsection*{Location privacy and geographical masks}

There is a vast literature on the problem of identity disclosure when dealing with spatially referenced data. The opportunities and challenges with regard to spatial data in the context of social sciences are discussed in great detail in~\cite{gutmann2007} and~\cite{gutmann2008}.

\smallskip

Articles~\cite{brownstein} and~\cite{curtis} give illustrative examples of how naive publishing of spatially referenced data can lead to a violation of anonymity: In both cases, the respective authors were able to reconstruct many of the original addresses successfully from published low resolution maps. A currently flourishing branch of research deals with anonymisation techniques for datasets containing mobility traces of individuals~\cite{gambs} (e.g., obtained via mobile phone tracking). This topic is usually referred to as \textit{location privacy}~\cite{krumm}.

\smallskip

In this article, however, we consider the deanonymisation risk that arises from the knowledge of the (approximate) distances between fixed spatial points assigned to the entities in a microdata table.
Various methods for the anonymisation of geographic point data (not necessarily taking additional covariates into consideration as in our case) have been discussed under the term of \textit{geographical masks}. \cite{armstrong} and~\cite{okeefe} provide comprehensive outlines of the existing methods.

\smallskip

A noteworthy method is due to Wieland et al.~\cite{wieland}, who developed a method based on linear programming that moves each point in the dataset as little as possible under a given quantitative risk of re-identification. However, the aim of nearly all proposed anonymisation techniques for spatially referenced data consists in distorting the spatial distribution with respect to the underlying geographical area as little as possible, whereas attempts predominantly focusing on the preservation of distances have not yet been discussed in the context of spatial data. It appears to be obvious that neglecting the underlying geographical area might yield more accurate results regarding distance calculations. 

\subsection*{Social network anonymisation}

The use of a graph model in this article might suggest a strong connection between our approach and the methods discussed in the area of \textit{social network anonymisation}~\cite{zheleva}. However, we model the microdata with known inter-record distances using a complete graph with vertex labels and edge weights, which is a very specific model in contrast to the more general graph models commonly used in social network analysis.

\smallskip

Indeed, the graphs modelling social networks are usually a long way off from being complete and their edges are not usually weighted. For example, in~\cite{chester} the underlying graph model considers discrete edge labels instead of real valued weights only.

\smallskip

Furthermore, active attacks (consisting in the addition of nodes to the published network by an intruder) as in~\cite{backstrom} do not seem to be sensible when investigating the risk of identity disclosure for published microdata. However, the active attack proposed in~\cite{backstrom} is related to the one in this paper because it also makes use of graph algorithmic building blocks. It consists in the detection of a subgraph in a larger graph, whereas the attack in this paper is based on finding the common subgraphs of two different graphs.

\subsection*{Pattern recognition} 
To the best of our knowledge, this paper is the first one to make use of a graph model for a microdata file and the distances between its records. Finding a matching between two such graph models constitutes the basic principle of the graph theoretic linkage attack proposed in this article and is an often considered problem in the \textit{pattern recognition} field and its various areas of application (see~\cite{conte} as a source providing an extensive outline).

\smallskip

Fundamental to our presentation is the article by Levi~\cite{levi}, which motivates to transform the problem of finding the (maximum) common subgraphs of two graphs into a (maximum) clique detection problem, and its adaption in~\cite{fober} where the original approach by Levi has been relaxed in order to deal with approximate common subgraphs as well. This transformation to the maximum clique detection problem is of particular interest due to its various fields of application (e.g. biochemistry~\cite{fober}). The problem of finding a maximum clique in a graph is known to be NP-hard~\cite{garey} and a great deal of attention has been paid to the development of techniques for solving this problem either exactly or at least approximately~\cite{bomze}.
For the simulation study in section~\ref{sec:experimental_results} of this paper, we made use of the maximum clique detection algorithm introduced by Konc and Jane\v{z}i\v{c} in~\cite{konc}. Exploring the limits of our approach in view of its scalability towards very large files is postponed to future research.

\section{A graph model for microdata in a metric space}\label{sec:graph_model}

\subsection*{Preliminaries}
A metric space is a pair $(X,d)$, where $X$ is a set and $d$ is a (distance) function $d: X \times X \to \mathbb R$ satisfying the following three conditions: (i) $d(x,x)=0$ and $d(x,y)>0$ whenever $x \neq y$, (ii) $d(x,y)=d(y,x)$ and (iii) $d(x,y)\leq d(x,z)+ d(z,y)$.

We assume that the deduplicated microdata table $T$ at hand contains information with respect to an attribute set $\mathcal A:=\{A_1,\ldots,A_m\}$ about $N_T \in \mathbb N$ entities from an underlying population. The fact that the distances between the entities of $T$ are known can be modelled in mathematical terms by means of a function $\tau: [N_T]:=\{1,\ldots,N_T\} \to (X,d)$, $i \mapsto \tau(i)$ which maps the $i$th record/entity of $T$ to a point $\tau(i)$ in a metric space $X$ such that the distance between records $i$ and $j$ of $T$ is equal to $d_{ij}:=d(\tau(i),\tau(j))$. The distances between all the entities can then be stored in the $N \times N$ distance matrix $D=(d_{ij})$. Such a pair $(T,D)$ is hereafter referred to as \textit{microdata in a metric space}.

Note that we did not state any assumptions on the function $\tau$ such as injectivity or surjectivity. It is easy to see that $[N_T]$ itself becomes a metric space by the pullback of $d$ via $\tau$ if and only if $\tau$ is injective (see page 81 in~\cite{deza}). In general, $[N_T]$ becomes a pseudometric space only. From our point of view, this flexibility regarding $\tau$ is intended as the records of a microdata table often only form a pseudometric instead of a metric space, which is illustrated by the following example:
Consider microdata about individuals which have been gathered in a scientific survey. If two respondents share a common residence, the geographical distance between these respondents will be equal to zero and thus the set of respondents with the related distances between them forms a pseudometric space only. Thus, the distance matrix $D$ is not assumed to be a proper distance matrix, i.e. zeroes outside the diagonal are permitted.

\subsection*{Some terms from graph theory} Given a set $S$, we denote the set of its two-element subsets by $[S]^2$. A \textit{(simple undirected) graph} $\mathcal G=(V,E)$ consists of a set $V$ (whose elements are termed \textit{vertices}) and a set $E\subseteq [V]^2$ of \textit{edges}. 
The cardinality $|V|$ of $V$ is called the order of $\mathcal G$.
Two distinct vertices $v$ and $w$ of $V$ are \textit{adjacent} if $\{v,w\} \in E$. The existence of an edge between $v$ and $w$ will sometimes be denoted by $vw \in E$ as a shorthand.
A graph is called \textit{complete} if any two of its vertices are adjacent.
A graph $\mathcal G'=(V',E')$ with $V' \subseteq V$ and $E'\subseteq [V']^2 \cap E$ is a \textit{subgraph} of $\mathcal G=(V,E)$. If $E' =[V']^2 \cap E$ holds, the graph $\mathcal G'$ is called an \textit{induced subgraph} of $\mathcal G$ or we say that the subset $V'$ of vertices induces $\mathcal G'$ in $\mathcal G$ which is denoted by $\mathcal G'=\mathcal G[V']$. A subset of the vertex set $V$ is a \textit{clique} if the subgraph induced by these vertices is complete. A clique containing $k$ elements is termed a $k$\textit{-clique}. A clique is \textit{maximal} if it is not contained in a larger clique. A clique is \textit{maximum} if there is no other clique containing more vertices. Clearly, a maximum clique is always maximal, but generally not vice versa. 
\par\smallskip
The notion of a vertex-labelled and edge-weighted graph is of fundamental importance to the graph model for microdata in a metric space introduced below. This notion is just a special case of the more general notion of an \textit{attributed graph} which is frequently used in the pattern recognition community~\cite{bunke_riesen}. 

\begin{definition}\label{def:graph}
Let $\mathcal L_V$ be a set of vertex labels.
A \textit{vertex-labelled} and \textit{edge-weighted graph} is a four-tuple $\mathcal G=(V,E,\lambda,\omega)$, where $V$ is the vertex set, $E \subseteq [V]^2$ the edge set, $\lambda: V \to \mathcal L_V$ the vertex-labelling function and $\omega: E \to \mathbb R$ a weight function which assigns real numbers to the edges.
\end{definition}

\subsection*{The graph model}
Let $(T,D)$ be microdata in a metric space and $N_T$ the number of records in $T$ as above. An associated vertex-labelled and edge-weighted graph $\mathcal G=\mathcal G(T,D)=(V,E,\lambda,\omega)$ can be defined as follows: Set $V=\{1,\ldots,N_T\}$, $E=[V]^2$ and define $\omega_E: E \to \mathbb R$ via $\omega_E(ij)=d_{ij}:=d(\tau(i),\tau(j))$; the labelling function $\lambda_V: V \to \mathcal L_V$ assigns a certain part of the information stored in $T$ for a record to the corresponding vertex of the graph $\mathcal G$ (see example~\ref{example:graph_model} below).
Note that the simple undirected graph $\mathcal G_{\text{simple}}:=(V,E)$ obtained from $\mathcal G$ by forgetting vertex labels and edge weights is the complete graph $K_{N_T}$ with $N_T$ vertices. This graph theoretical structure appears adequate for modelling microdata in a metric space: Loops, i.e. edges linking a vertex with itself, are not necessary because $d_{ii}=0$ for any vertex $i \in V$ and undirected edges are sufficient for reflecting the distance from the corresponding edge weights due to the symmetry $d_{ij}=d_{ji}$ of the distance matrix $D=(d_{ij})$. Obviously, it would be easy to widen this model, e.g. by introducing directed edges, if this were necessary for a specific application.

\begin{example}\label{example:graph_model}
Consider the imaginary microdata provided by table~\ref{table:example_microdata} containing personal microdata with respect to the attributes \texttt{name}, \texttt{sex}, \texttt{birth location} and \texttt{year of birth}. The function $\tau$ maps each individual to the geographic coordinates (longitude $\lambda$ and latitude $\theta$ in degrees) of the correspoding birth location with respect to the World Geographic System WGS 84, i.e.
\begin{align*}
\tau(1)&=(-0.1198244,51.51121) \quad \text{(Alice was born in London)}\\
\tau(2)&=(2.3522219,48.85661) \quad \text{(Bob was born in Paris)}\\
\tau(3)&=(-3.7037902,40.41678) \quad \text{(Eve was born in Madrid)}\\
\tau(4)&=(13.4049540,52.52001) \quad \text{(Walter was born in Berlin)}
\end{align*}

Assuming a spherical shape with radius $R=6371\text{ km}$ for the earth and converting degrees to radians, the geographical distance $d$ between two locations $(\lambda_1,\theta_1)$, $(\lambda_2,\theta_2)$ can be calculated as $d=R \cdot \phi$ where
$$\cos \phi = \sin \theta_1 \sin \theta_2 + \cos \theta_1 \cos \theta_2 \cos(\lambda_1-\lambda_2).$$
Using this formula leads to the following distance matrix $D$:

$$D=(d_{ij})=\begin{pmatrix}
0 & 343.6 & 1264.0 & 930.9\\ 
343.6 & 0 & 1052.9 & 877.5\\ 
1264.0 & 1052.9 & 0 & 1869.1\\
930.9 & 877.5 & 1869.1 & 0
\end{pmatrix}.$$

The corresponding graph model is then given by $V=\{1,2,3,4\}$, $E=[V]^2$ and the edge weights are defined via $\omega(ij)=d_{ij}=d_{ji}$. We define the vertex labelling function by assigning the information regarding the attributes \texttt{sex} and \texttt{year of birth} to each vertex, i.e. formally, we have $\lambda_V: V \to \text{dom}(\texttt{sex}) \times \text{dom}(\texttt{yob}) $.

The resulting vertex-labelled and edge-weighted graph can be visualised as in figure \ref{fig:example_model}.

\begin{minipage}{\textwidth}
  \begin{minipage}[b]{0.48\textwidth}
    \centering
    \begin{tabular}{c|c|c|c}
    name & sex & birth location & year of birth\\ 
    \hline
    Alice & f & London & 1978\\ 
    Bob & m & Paris & 1965\\ 
    Eve & f & Madrid & 1943\\
    Walter & m & Berlin & 1931\\
    \end{tabular}
    \vspace{2cm}
    \captionof{table}{Example microdata table. The table contains the attributes \texttt{name}, \texttt{sex}, \texttt{birth location} and \texttt{year of birth}.}\label{table:example_microdata}
    \end{minipage}
\begin{minipage}[b]{0.03\textwidth}
\phantom{x}
\end{minipage}
      \begin{minipage}[b]{0.48\textwidth}
        \centering
        \includegraphics[width=0.8\textwidth]{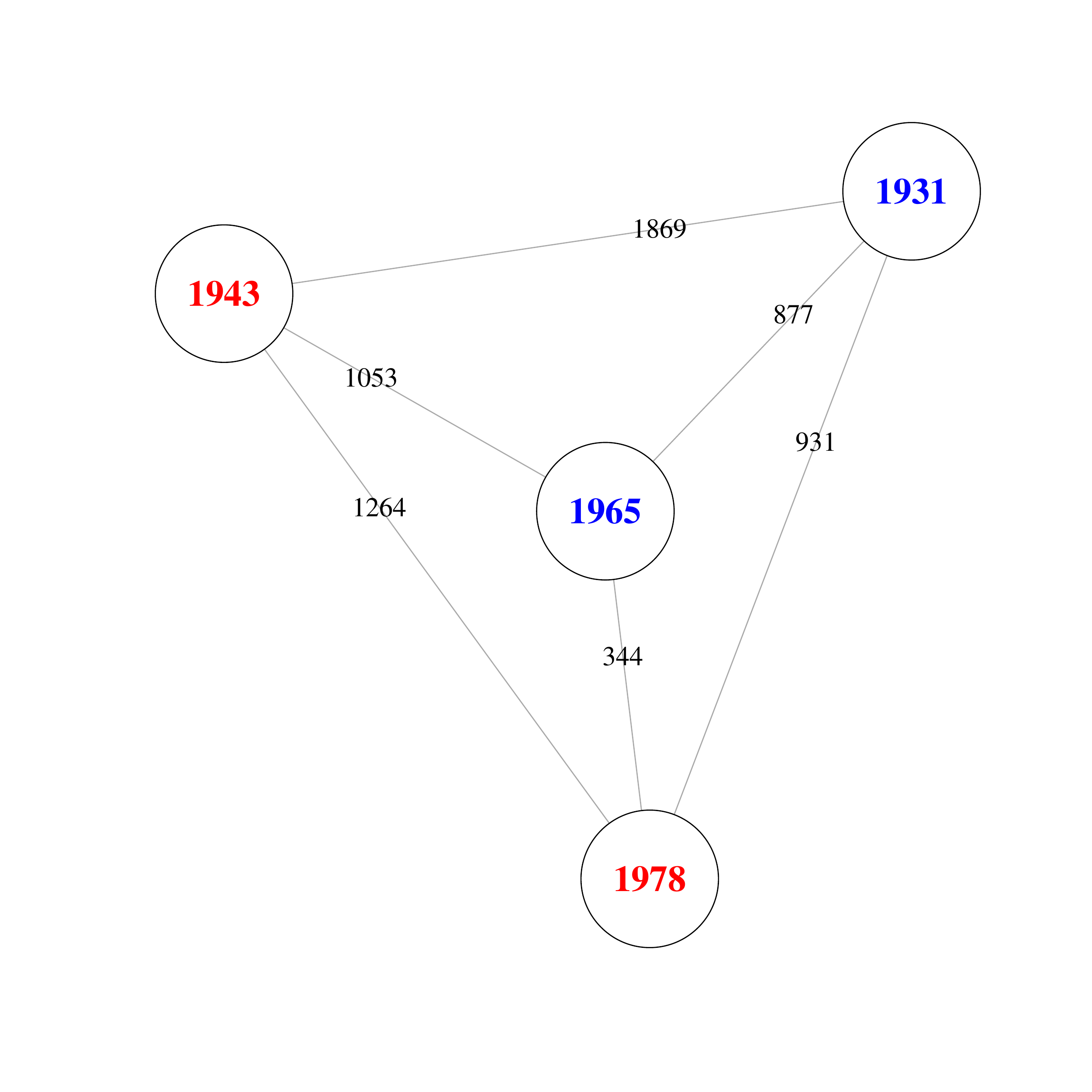}
        \captionof{figure}{The graph model for the example microdata. The attribute \texttt{sex} is indicated by the colour of the vertex labels.}\label{fig:example_model}
      \end{minipage}
  \end{minipage}
\end{example}

\section{A graph theoretic linkage attack}\label{sec:linkage_attack}

\subsection*{Prerequisites for the attack} 

In order to make any kind of linkage attack with the objective of identity disclosure, we have to at least presuppose that an appropriate external microdata file is available to the data snooper.

\begin{assumption}\label{assumption:identification_file}
The snooper is in possession of an identification file containing direct identifiers.
\end{assumption}

Under this assumption, classical linkage attacks based on comparisons considering the quasi-identifiers of the identification and target file can be conducted.
As already mentioned in section \ref{sec:background}, in the literature on the deanonymisation of microdata, these represent an important mode of attack aimed at identity disclosure. In order to perform a linkage attack that goes beyond the ordinary ones described above by also taking the information given by the pairwise distances between the records into consideration, we have to expand the setup by a second assumption.

\begin{assumption}\label{assumption:compute_distances}
The snooper is able to calculate the distances between the entities in the identification file at least approximately.
\end{assumption}

Although in some cases assumption \ref{assumption:compute_distances} might not be fulfilled, it is easy to find examples of when this would indeed be the case. For instance, when the target file containing survey data is enriched by the geographic distances between the respondents' residences, we assume that the snooper can geocode the addresses of the individuals in the identification file and calculate the corresponding distance matrix. In this example, there will be some dependence on the methods used for geocoding and distance calculation, a fact which has to be considered in the creation of an attack mode. Analogously, any modification of the distances in the target file to be carried out by the data holder for the purpose of anonymisation will have to be taken into consideration.

\subsection*{Approximate common subgraphs}
Due to assumptions \ref{assumption:identification_file} and \ref{assumption:compute_distances}, a data snooper can create a vertex-labelled and edge-weighted graph as defined in section \ref{sec:graph_model} for both the target and identification file. At this step, the snooper will only consider the common quasi-identifiers of both files for the definition of the vertex labels because a comparison of records can only be based on such attributes. Hereafter, the resulting graphs will be referred to as the \textit{target} and \textit{identification graph}.

\smallskip

Hence, classical linkage attacks consist in trying to find vertices in the target graph for each vertex in the identification graph that result in matches for the accompanying vertex labels. In the parlance of graph theory, this approach is equivalent to the search for \textit{common subgraphs} of order $1$, a notion which will be made precise below. This course of action will usually (e.g., if the target file satisfies $k$-anonymity for some $k>1$) lead to ties, that cannot be broken without extra information.

\smallskip

However, due to the additional information given by the edge weights in the graph model, the snooper is able to search for complete common subgraphs of order $>1$, which forms the essence of our attack.
It is intuitively apparent that taking edge weights into consideration increases a snooper's chances of evaluating the credibility of potential matches. 
For instance, if we consider vertices $v_1,v_2$ in the target graph $\mathcal G_1=(V,E,\lambda_V,\omega_E)$ and $w_1,w_2$ in the identification graph $\mathcal G_2=(W,F,\lambda_W,\omega_F)$ such that $\lambda_V(v_1)=\lambda_W(w_1)$ and $\lambda_V(v_2)=\lambda_W(w_2)$, we observe coincidence regarding the vertex labels.
If the corresponding edge weights $\omega_E(v_1v_2)$ and $\omega_F(w_1w_2)$ are at least approximately equal (denoted by $\omega_E(v_1v_2)\approx \omega_F(w_1w_2)$), this fact will augment the credibility of the two matches $(v_1,w_1)$ and $(v_2,w_2)$. Conversely, a large distortion with respect to the corresponding edge weights will reduce this credibility: In this case, at least one of the considered matches should be false. These considerations can easily be generalised to more than two matches and all accompanying edge weights. The more potential matches preserve all the accompanying edge weights, the more the credibility of all these potential matches will increase. This motivates the snooper to identify nearly identical substructures in both graphs which are as large as possible.

\smallskip

As indicated above, it seems convenient to allow some deviation with respect to the edge weights in this context due to deviations which cannot be circumvented by a snooper (as mentioned in the special case of geographic distances above). All of these considerations can be dealt with rigorously using the notion of an \textit{approximate common subgraph} of two vertex-labelled and edge-weighted graphs. This notion is made precise by means of the following definition:

\begin{definition}\label{def:common_subgraph}
Let $\mathcal G_1=(V,E,\lambda_
V,\omega_E)$ and $\mathcal G_2=(W,F,\lambda_W,\omega_F)$ be two vertex-labelled and edge-weighted graphs in the sense of definition~\ref{def:graph}. An \textit{approximate common subgraph} of $\mathcal G_1$ and $\mathcal G_2$ is given by subsets $S \subseteq V$, $T \subseteq W$ and a bijection $\varphi: S \to T$ such that the following two statements are true:
\begin{enumerate}[(i)]
\item $\lambda_V(s)=\lambda_W(\varphi(s))$ for all $s \in S$.
\item For all $s_1,s_2 \in S$ we have either
\vspace{-0.25em}
\begin{enumerate}[(a)]
\item $s_1s_2 \in E$, $\varphi(s_1)\varphi(s_2)\in F$ and $\omega_E(s_1s_2) \approx \omega_F(\varphi(s_1)\varphi(s_2))$, \quad or
\item $s_1s_2 \notin E$ and $\varphi(s_1)\varphi(s_2) \notin F$.
\end{enumerate}
\end{enumerate}
\end{definition}

Condition (ii) of definition~\ref{def:common_subgraph} guarantees that vertices $v_1,v_2 \in V$ can only be mapped to vertices $w_1, w_2 \in W$ if either both pairs of vertices are adjacent or non-adjacent (the non-adjacency even yields a sufficient condition for making such a mapping possible). Since we only consider complete graphs in this article, only requirement (a) in condition (ii) has to be checked because requirement (b) will never be fulfilled.
Furthermore, the interpretation of $\approx$ in definition~\ref{def:common_subgraph} has to be made precise depending on the prevailing situation and especially on the possible perturbations of the distances caused by the data holder before publishing the microdata. This issue will be dealt with in detail in example~\ref{example:poets} in this section and the simulation study in section~\ref{sec:experimental_results}. It would have certainly been possible to allow some amount of deviation regarding the vertex labels as well by introducing a similarity measure on the set of vertex labels. In this paper, however, we do not deal with this aspect. We require exact coincidence for the labels of two vertices to be matched since we are primarily interested in the effect of how publishing (perturbed) distances influences the risk of identity disclosure.

\subsection*{The product graph} In order to tackle the problem of finding approximate common subgraphs of two vertex-labelled and edge-weighted graphs $\mathcal G_1$ and $\mathcal G_2$, we transform this problem to a clique detection problem in an appropriately defined simple undirected graph $\mathcal G_\otimes$, the product graph of $\mathcal G_1$ and $\mathcal G_2$.

\begin{definition}\label{def:product_graph}
Let $\mathcal G_1=(V,E,\lambda_V,\omega_E)$ and $\mathcal G_2=(W,F,\lambda_W,\omega_F)$ be two vertex-labelled and edge-weighted graphs as in definition \ref{def:graph}. The \textit{product graph} $\mathcal G_\otimes=(V_\otimes,E_\otimes)$ of $\mathcal G_1$ and $\mathcal G_2$ is a simple undirected graph defined through
\begin{align*}
V_\otimes &= \{(v,w)\in V \times W : \lambda_V(v)=\lambda_W(w)\} \quad \text{ and }\\
E_\otimes &= \bigg \{ \{(v_1,w_1),(v_2,w_2) \} : v_1\neq v_2, w_1\neq w_2 \text{ and either (a) }v_1v_2 \in E, w_1w_2 \in F \text{ and} \\
&\hspace{0.9cm}\omega_E(v_1v_2) \approx \omega_F(w_1w_2) \text{, or (b) }v_1v_2 \notin E \text{ and } w_1w_2 \notin F\bigg\}.
\end{align*}
\end{definition}

The announced transformation of the maximum approximate common subgraph problem into the maximum clique problem is achieved via the following theorem:
\begin{theorem}\label{theorem}
Consider the setup of definition \ref{def:product_graph}. There is a one-to-one correspondence between the approximate common subgraphs of order $k$ and $k$-cliques of $\mathcal G_\otimes$.
\end{theorem}

\begin{proof}
Let an approximate common subgraph of $\mathcal G_1$ and $\mathcal G_2$ of order $k$ be given by the vertex sets $S=\{v_1,\ldots,v_k\}\subseteq V$ and $T=\{w_1,\ldots,w_k\} \subseteq W$, respectively. Without loss of generality, we assume $\varphi(v_i)=w_i$ for $i\in \{1,\ldots,k\}$ under the corresponding subgraph isomorphism $\varphi$. Condition (i) in definition \ref{def:common_subgraph} yields $(v_i,w_i) \in V_\otimes$ for $i=1, \ldots, k$. Moreover, for distinct $i, j \in \{1,\ldots,k\}$ we have $v_iv_j \in E \Leftrightarrow w_iw_j\in F$.
If $v_iv_j \in E$ condition (ii)  in definition \ref{def:common_subgraph} implies that $\omega_E(v_iv_j) \approx \omega_F(\varphi(v_i)\varphi(v_j))=\omega_F(w_iw_j)$ and $(v_i,w_i)$ and $(v_j,w_j)$ are adjacent in $\mathcal G_\otimes$.
Because $i,j$ were chosen arbitrarily, $\mathcal C:=\{(v_1,w_1),\ldots,(v_k,w_k)\}$ forms a $k$-clique in $\mathcal G_\otimes$.
\par\medskip
Conversely, let $\mathcal C$ be a $k$-clique in $\mathcal G_\otimes$ given by vertices $(v_1,w_1),\ldots,(v_k,w_k) \in V_\otimes$. We define $S=\{v_1,\ldots,v_k\}$, $T=\{w_1,\ldots,w_k\}$ and $\varphi : S \to T$ via $\varphi(v_i)=w_i$. Then $\varphi$ is a bijection and we obtain
$\lambda_V(v_i)=\lambda_W(w_i)=\lambda_W(\varphi(v_i)) \text{ for } i=1,\ldots,k$. Thus, condition (i) in definition \ref{def:common_subgraph} is satisfied. The validity of the second condition follows from the fact that either $v_iv_j \notin E$ and $w_iw_j \notin F$ or $v_iv_j \in E$ and $w_iw_j \in F$ and that we have $\omega_E(v_iv_j) \approx \omega_F(w_iw_j) = \omega_F(\varphi(v_i)\varphi(v_j))$ in the latter case.
\end{proof}

\begin{corollary}
The problem of finding a maximum approximate common subgraph of two vertex-labelled and edge-weighted graphs is equivalent to the problem of detecting a maximum clique in the associated product graph.
\end{corollary}

Before putting all the ingredients discussed so far together, we want to make some remarks regarding theorem \ref{theorem}, which is folklore within the pattern recognition community. The first time that the correspondence between common subgraphs of two graphs and the cliques in the corresponding product graph was considered was in~\cite{levi} for exact isomorphisms between simple graphs including vertex labels. Since then, this approach has become a standard tool for tackling graph matching problems in various fields of application~\cite{conte}.

The definition of the product graph recently presented in~\cite{fober} is equivalent to the one we use in this paper. However, in that paper the authors relax the concept of a clique (which appears to be too restrictive for their application) to the less restrictive concept of a $\gamma$-quasi-clique and propose an algorithm for quasi-clique detection based on local clique merging.

\smallskip

For our application only requirement (a) in the construction of the edge set $E_\otimes$ is relevant since we deal with complete graphs only. The condition that $v_1\neq v_2 \text{ and } w_1\neq w_2$ in the construction of $E_\otimes$ guarantees that the cliques in the product graph correspond to one-to-one matches between the vertices of the two graphs to be matched.
For complete graphs without loops, this condition is implicitly part of the requirement that $\omega_E(v_1v_2) \approx \omega_F(w_1w_2)$. If we had permitted loops in our graph model and assigned a weight of $0$ to each of these loops, $v_1\neq v_2 \text{ and } w_1\neq w_2$ would have been necessary in order to keep the statement of theorem \ref{theorem} true, which is illustrated by the following example:

\begin{example}
Consider $\mathcal G_1=(V,E,\lambda_V,\omega_E)$ given by $V=\{1\}$. Then, necessarily $E=\emptyset$ and $\omega_E$ is redundant. We consider the product graph $\mathcal G_\otimes$ of $\mathcal G_1$ with $\mathcal G_2=(W,F,\lambda_W,\omega_F)$ where $W=\{ 2,3\}$, $F=\big\{\{2,3\} \big\}$ and $\omega_F \big ( \{2,3\} \big )=\frac{\varepsilon}{2}$ for some $\varepsilon>0$. Moreover, we assume that $\lambda_V(1)=\lambda_W(2)=\lambda_W(3)$. Thus, the vertex set $V_\otimes$ of the product graph is given by $V_\otimes=\{(1,2),(1,3)\}$. Let us define that two edge weights are approximately the same if the absolute value of their difference is less than $\varepsilon$. Then, according to our definition of the edge set $E_\otimes$ the two vertices from $V_\otimes$ are not adjacent because they coincide in their first component. Without the condition $v_1\neq v_2 \text{ and } w_1\neq w_2$ and allowing loops of zero weight in the graph models, however, these two vertices would have been adjacent. Obviously, in this case the resulting clique $\mathcal C=V_\otimes$ would not be related to an approximate common subgraph of $\mathcal G_1$ and $\mathcal G_2$.
\end{example}

\subsection*{Overview}

Let us now formulate the overall graph theoretic linkage attack.
\begin{center}
\fbox{\parbox{0.98\linewidth}{
\textbf{Graph Theoretic Linkage Attack on Microdata in a Metric Space}
\par\medskip
\begin{tabular}{ll}
INPUT & Target data $(T_1,D_1)$, identification data $(T_2,D_2)$\\
OUTPUT & List of matches between records from $T_1$ and $T_2$
\end{tabular}

\begin{compactenum}
\item Build target graph $\mathcal G_1$ from $(T_1,D_1)$.
\item Build identification graph $\mathcal G_2$ from $(T_2,D_2)$ (possible under assumptions \ref{assumption:identification_file} and \ref{assumption:compute_distances}).
\item Build product graph $\mathcal G_\otimes$ (requires reasonable definition of $\approx$).
\item Find a maximum clique $\mathcal C_{\text{max}}$ in $\mathcal G_\otimes$ (using some maximum clique detection algorithm).
\item Extract matches from $\mathcal C_{\text{max}}$.
\end{compactenum}
}}
\end{center}

Let us make a brief comment on step 4 of the attack: As already indicated in section~\ref{sec:background}, there is a vast literature concerning the problem of maximum clique detection in graphs. A systematic comparison of the prevalent techniques to tackle this problem in the context of our application goes beyond the scope of this paper and is postponed to future research.

To conclude this section, we illustrate the process of the proposed graph theoretic linkage attack using a small-scale example which makes use of the data summarised in appendix \ref{sec:example_data}.

\begin{example}\label{example:poets}
Consider microdata table~\ref{table:poets} in appendix \ref{sec:example_data}, which contains information about various important European poets. This table is anonymised by removing the direct identifier \texttt{name}, generalising the attribute \texttt{yob} (year of birth) to \texttt{cob} (century of birth) and removing the information about the birth location (\texttt{loc}). The attribute \texttt{language} remains unchanged. This yields anonymised table~\ref{table:poets_anonymized} in appendix \ref{sec:example_data}.

\smallskip

While the spatial information \texttt{loc} has been deleted from this table, the distance matrix $D_1$ (see appendix~\ref{sec:example_data}) containing the geographic distances between the birth locations is meant to be published in addition to table~\ref{table:poets_anonymized}. We assume that the snooper is in possession of the identification microdata in table~\ref{table:poets_intruder}, i.e. the attributes \texttt{cob} and \texttt{language} serve as quasi-identifiers. By geocoding the birth locations and calculating the geographic distances, the snooper obtains the distance matrix $D_2$. Graph models $\mathcal G_1$ and $\mathcal G_2$ for the target and identification data can be built by using this information and are visualised in figure \ref{fig:example_poets_models}.

\smallskip

Table \ref{table:vertex_matches} lists all the possible matches if the snooper takes only the vertex labels into consideration. These eleven matches form the vertex set of the product graph as well. Note that this set would already constitute the final outcome of a linkage attack where the distances are not taken into consideration.
\begin{figure}[!h]
\begin{minipage}[b]{0.5\linewidth}
\centering
\includegraphics[width=0.99\textwidth]{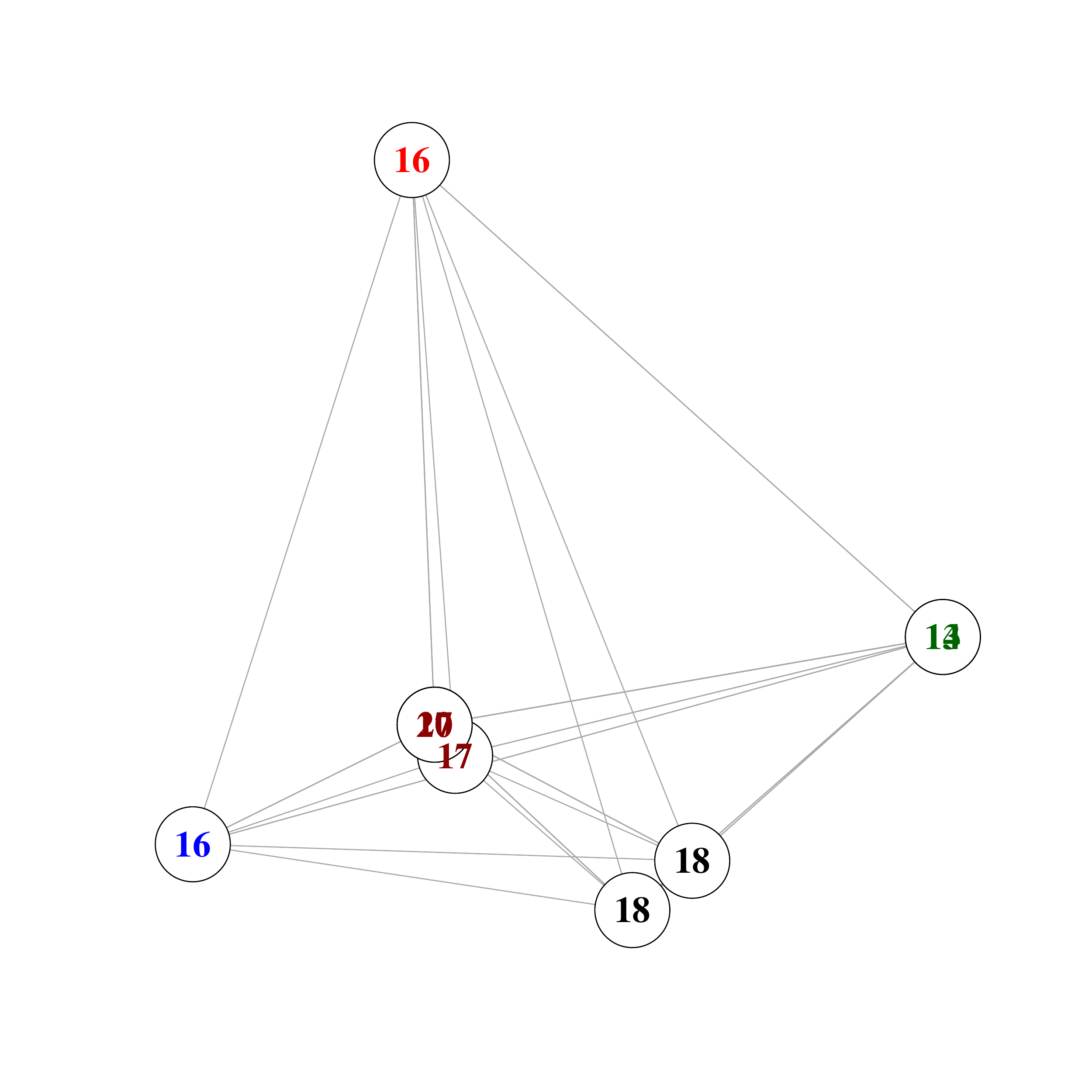}
\end{minipage}
\begin{minipage}[b]{0.5\linewidth}
\centering
\includegraphics[width=0.99\textwidth]{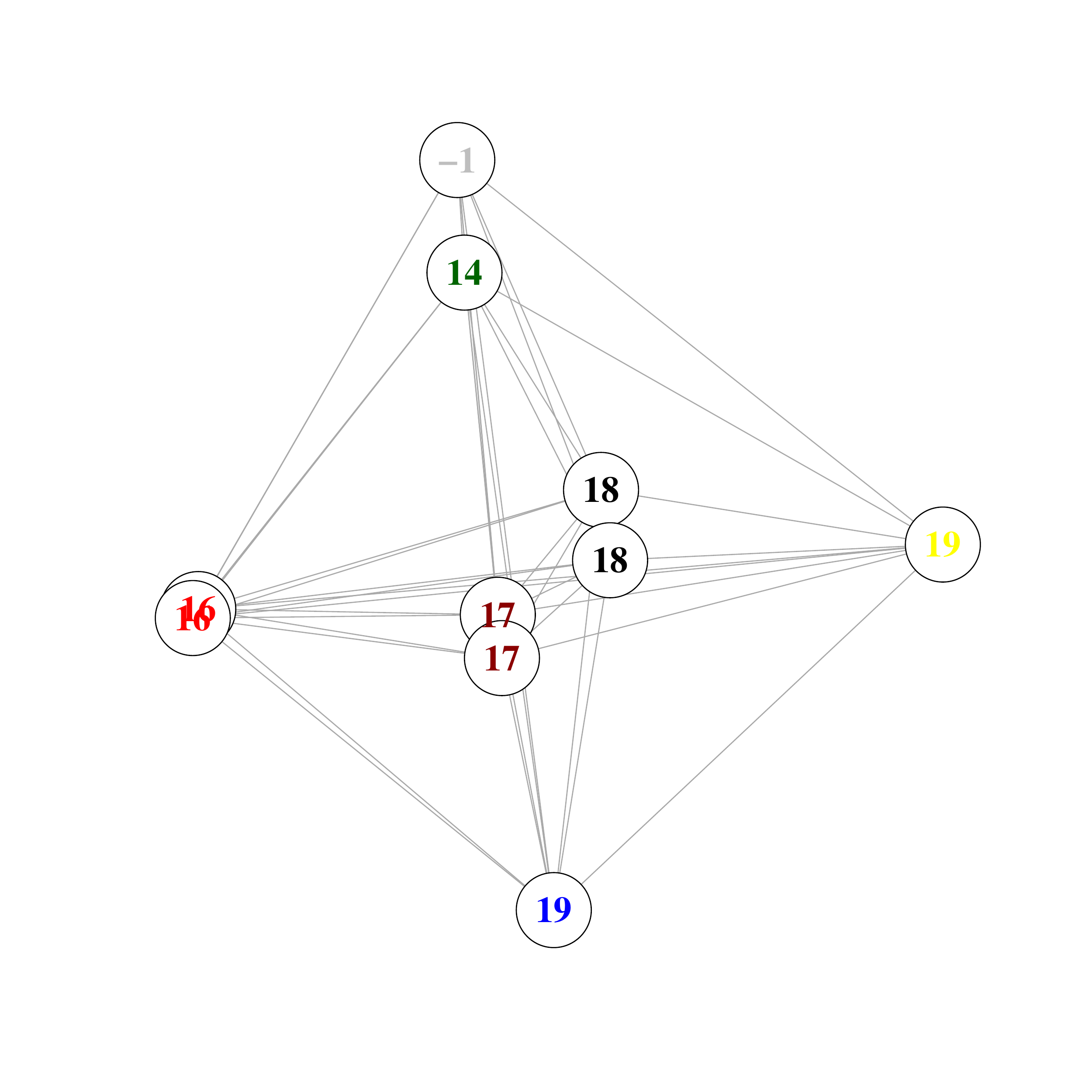}
\end{minipage}
\caption{Graph models for target (left) and identification (right) microdata in example \ref{example:poets}. The layout of the graphs was chosen such that the edge lengths give an approximate indication of the distances. The attribute \texttt{language} is indicated by the vertex label colour, whereas the attribute \texttt{cob} is indicated by the vertex label itself.}\label{fig:example_poets_models}
\end{figure}
\begin{table}[!h]
\centering
\begin{tabular}{ccc}
  \hline
vertex of product graph & rownumber target file & rownumber identification file\\ 
  \hline
1 &   1 &   1 \\ 
  2 &   2 &   2 \\ 
  3 &   3 &   3 \\ 
  4 &   6 &   3 \\ 
  5 &   4 &   4 \\ 
  6 &   7 &   4 \\ 
  7 &   3 &   6 \\ 
  8 &   6 &   6 \\ 
  9 &   4 &   7 \\ 
  10 &   7 &   7 \\ 
  11 &   2 &   9 \\ 
   \hline
\end{tabular}
\caption{Possible matches between tables~\ref{table:poets_anonymized} and~\ref{table:poets_intruder} with respect to the quasi-identifiers \texttt{cob} and \texttt{language} only, i.e. vertex labels in the accompanying graph models.}\label{table:vertex_matches}
\end{table}

For the construction of the product graph, we allow an absolute deviation of five kilometers with respect to the edge weights, i.e. we define $\omega_E(v_1v_2) \approx \omega_F(w_1w_2) \Leftrightarrow |\omega_E(v_1v_2)-\omega_F(w_1w_2)|<5$.\footnote{We previously mentioned that allowing such a deviation is already necessary because of errors that appear due to the fact that the data holder and snooper generally use different methods for geocoding and distance computation. This fact was also addressed in this example by geocoding the birth locations of the target microdata via Wikipedia and the birth locations of the identification file by means of the command \texttt{geocode} provided by the R package \texttt{ggmap}.} This definition of $\approx$ leads to the product graph shown in figure~\ref{fig:product_graph}.

\begin{figure}[]
  \centering
  \includegraphics*[width=0.5\textwidth]{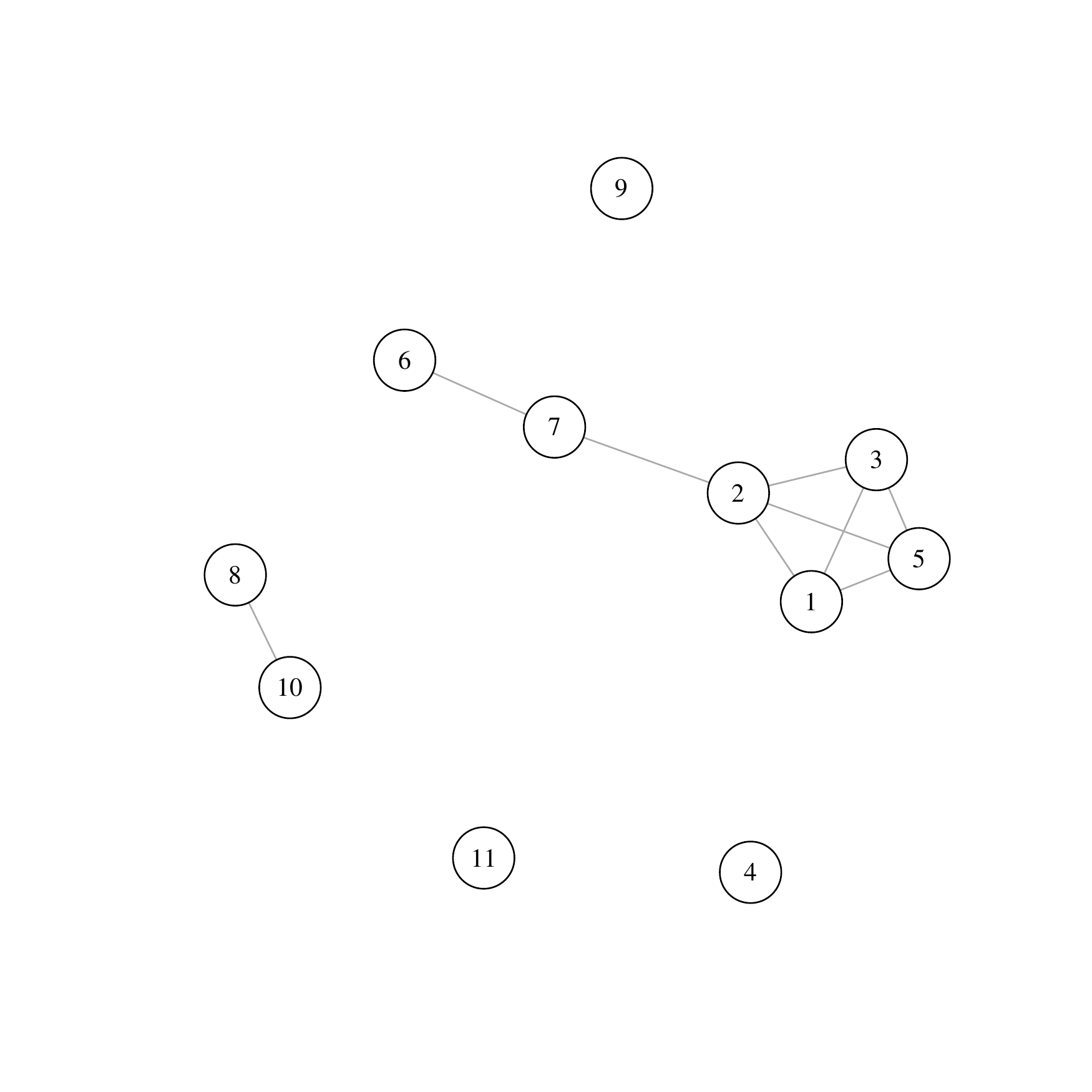}
  \caption{Product graph in example~\ref{example:poets}.}\label{fig:product_graph}
\end{figure}

As can be easily seen from figure~\ref{fig:product_graph}, the product graph contains a unique maximum clique $\mathcal C:=\{1,2,3,5\}$. Therefore, a snooper following the protocol of the graph theoretic linkage attack would accept the potential matches in rows 1,2,3 and 5 in table \ref{table:vertex_matches} as matches and reject the remaining ones.
\end{example}

Although example~\ref{example:poets} is artificial, it illustrates some of the phenomena that also appear when real-world data are taken into consideration:
\begin{itemize}
\item The definition of $\approx$ has to be chosen carefully. In the present example, distances between cities scattered all over the European continent are considered so that even the rather rough definition above (allowing for an absolute deviation of five kilometers) will yield a useful result. In general, the definition of $\approx$ has to be chosen such that as many common edges as possible of the target and identification graph are detected correctly, i.e. not classifying too many edges as approximately the same that are actually different. The definition of $\approx$ will be studied in greater detail as part of the simulation study in section~\ref{sec:experimental_results}.
\item A successful match of the respective first records of both tables would have already been possible unambiguously without the additional distance information because both records are unique in their tables with respect to the corresponding quasi-identifiers. Nevertheless, using the additional distance information increases the credibility for this specific matching, which is now not only supported by the coincidence of the quasi-identifiers but also by the coincidence of the distances to other matches.
\item However, in certain cases unambiguous matching is only possible because of the additional information about the distances. For example, record 3 of the target table could be matched with records 3 and 6 of the identification table only by taking the quasi-identifiers into consideration. This tie is resolved in our example by the extra information given by the edge weights.
\item Evidently, in practise there will be ties in the data that cannot be resolved by our method either. In our example, the records 9 and 10 of the target file do not differ according to their quasi-identifiers, however, they also cannot be distinguished by considering the distances to these records because the corresponding point locations (\texttt{loc}=Paris in both cases) coincide.
\item Finally, the attack has reduced the number of matches from eleven in table \ref{table:vertex_matches} to four. These matches indeed correspond to the actual overlap of the target and identification file.
\end{itemize}

Our toy example has shown that publishing inter-record distances might increase the risk of identity disclosure for microdata files. We confirm this result in the following section by investigating the effect of random noise addition to the input coordinates, which is a standard technique for the anonymisation of spatial point data.

\section{Experimental results}\label{sec:experimental_results}

\subsection*{Data} The data for the simulation study were generated as follows: In the first step, addresses from the German telephone book were sampled at random. Subsequently, geographic latitudes and longitudes based on the World Geodetic System 1984 were assigned to these addresses using the \texttt{geocode} command from the R package \texttt{ggmap}~\cite{kahle}. Finally, the geographic distances between the addresses were calculated to obtain the corresponding distance matrix.

\smallskip

We randomly assigned the points of the resulting metric spaces to example microdata containing (besides an ID) attributes concerning gender and age, which served as quasi-identifiers in our experiments. The attribute values were sampled in accordance with the actual distribution of these attributes derived during the German census 2011.\footnote{These demographic statistics can be downloaded from \url{https://ergebnisse.zensus2011.de/auswertungsdb/download?pdf=00&tableId=BEV_1_1_1&locale=DE}.}

\smallskip

 We generated data where both the size $N_1$ of the target and $N_2$ of the identification file where equal to $500$. The overlap $N_{\text{common}}$ of common records was chosen equal to $50$. The target and the identification file are visualised in figure~\ref{fig:files}. Note that the classification with respect to age (eleven age intervals) is rather rough; this guaranteed the existence of many duplicates with respect to the quasi-identifiers in our test microdata, which would result in ties when performing a classical linkage attack. Indeed, the order of the resulting product graph was equal to $|V_\otimes|=15517$.

\begin{figure}[]
\centering
\includegraphics[width=0.95\linewidth]{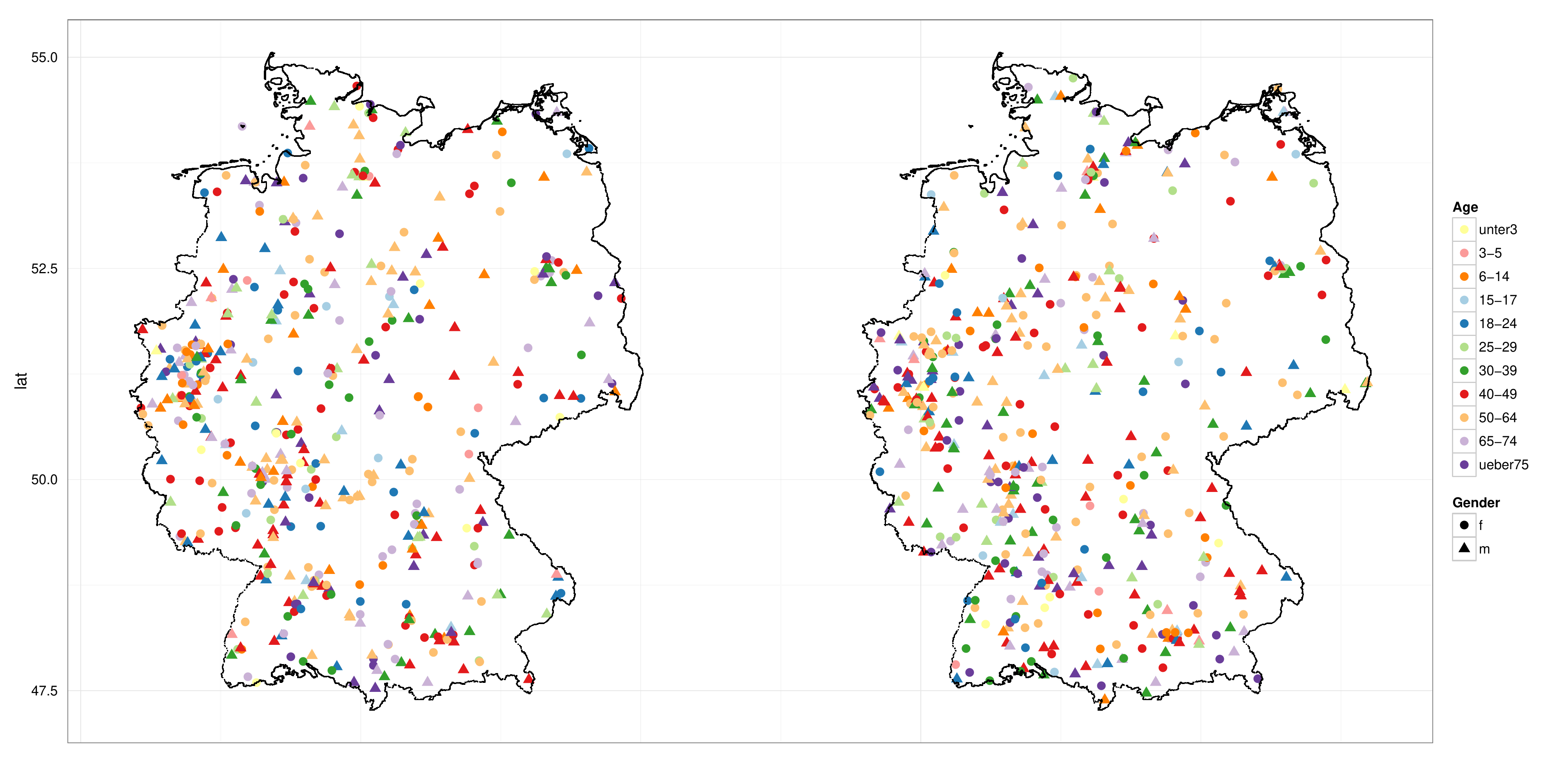}
\caption{Test data for the simulation study: Both the target (left) and identification file (right) contain $500$ records of which $50$ are common. Quasi-identifiers (\texttt{age}, \texttt{gender}) were sampled according to the actual distribution of these attributes due to the demographic statistics derived during the German census in 2011. Without the additionally released approximate distances the target file can be regarded as sufficiently anonymised: a classical linkage attack leads to a set of $15517$ potential matches between the target and identification file.}
\label{fig:files}
\end{figure}

\subsection*{Perturbation technique} A standard technique for the anonymisation of spatial point data consists in the addition of random noise to their coordinates (see section 3.2 in~\cite{armstrong}). In this section, we consider the performance of the proposed graph theoretical linkage attack under this anonymisation technique. To be more precise, $\mathcal N(0,\sigma^2)$-distributed Gaussian noise was added to the input coordinates of the target file before the distance matrix was calculated. Different instances of the standard deviation $\sigma$ were considered.

\subsection*{Fine-tuning of the attack} A suitable definition for the relation $\approx$ has to be found for the generation of the product graph in the graph theoretical linkage attack. Following Kerckhoffs' principle~\cite{petitcolas} (which implies that the security of a cryptosystem/anonymisation technique must not depend on the concealment of the algorithm in use), we assume that the data snooper knows that Gaussian noise is added to the geographic coordinates before the distances are calculated and, furthermore, that the standard deviation $\sigma$ is known to him (the latter assumption is in conformance with~\cite{armstrong}, who emphasise that \textit{all useful spatial analyses of masked data require some knowledge about the characteristics of the mask used}).

\smallskip

Under the assumption of a Euclidean distance function, the effect of random perturbation of the input coordinates on the squared distances can be studied theoretically, an approach which has been considered in~\cite{kenthapadi}. Such a rigorous mathematical analysis appears to be more difficult in the case of geographical distances.
For this reason, we assume that the snooper performs a little simulation study by which she/he investigates the effect of perturbation by Gaussian noise to the calculation of distances. To imitate this course of action, we sampled 1000 pairs of points from the area of the Federal Republic of Germany for each considered value of $\sigma$ and compared the distances before and after addition of Gaussian noise. Several sample quantiles and the sample variance (the latter will only be used for the evaluation of our experiments) of the deviation of the distances (which is defined as $d-d'$ where $d$ denotes the original distance and $d'$ the distance after perturbation) have been gathered and are recorded in table~\ref{table:distance_deviation}.

\smallskip

We use the empirical quantiles to define the interpretation of $\approx$: For a threshold parameter $\alpha \in (0,1)$, we define that two edge weights satisfy the relation $\approx$ if the corresponding deviation is greater than the empirical $\frac{1-\alpha}{2}$-quantile and smaller than the $\frac{1+\alpha}{2}$-quantile for the current value of $\sigma$. In this case, the distances from the identification file take on the role of $d$ and the distances from the target file the one of $d'$. The threshold parameter $\alpha$ chosen by the snooper is supposed to guarantee that a common edge of the target and identification graph is detected by the snooper with probability approximately equal to $\alpha$. Its effect will also be considered within this section.

\begin{table}[]
\centering\footnotesize
\begin{tabular}{rrrrrrrrr}
  \hline
$\sigma$ & 0.05 & 0.1 & 0.25 & 0.5 & 0.75 & 0.9 & 0.95 & sample variance \\ 
  \hline
0.005 & -1.1088 & -0.8558 & -0.4261 & 0.0226 & 0.4496 & 0.9256 & 1.2192 & 0.4799 \\ 
  0.010 & -2.3909 & -1.7191 & -0.8750 & 0.0798 & 0.9733 & 1.8218 & 2.2642 & 1.9677 \\ 
  0.015 & -3.3063 & -2.4633 & -1.2288 & 0.0810 & 1.4714 & 2.7492 & 3.4624 & 4.3312 \\ 
  0.020 & -4.6132 & -3.6261 & -2.0787 & -0.0615 & 1.7278 & 3.4013 & 4.3512 & 7.8732 \\ 
  0.025 & -5.6147 & -4.2924 & -2.2826 & -0.1592 & 2.2177 & 4.1254 & 5.3089 & 11.5378 \\ 
  0.030 & -6.4952 & -4.9210 & -2.5024 & 0.1763 & 2.9190 & 5.2730 & 6.6511 & 16.1313 \\ 
  0.035 & -8.3848 & -6.2351 & -3.0673 & -0.0665 & 3.0206 & 6.0428 & 7.9411 & 23.6334 \\ 
  0.040 & -9.1530 & -6.7866 & -3.7315 & -0.1884 & 3.4698 & 7.0085 & 8.7236 & 30.2768 \\ 
  0.045 & -11.0830 & -8.2160 & -4.1860 & -0.0680 & 3.6953 & 7.6620 & 10.2638 & 39.6836 \\ 
  0.050 & -11.4906 & -8.9544 & -4.7386 & -0.0057 & 4.5955 & 8.6728 & 11.4998 & 48.8299 \\ 
   \hline
\end{tabular}
\caption{Sample quantiles and variance of the considered distance deviation $d-d'$ for different values of $\sigma$.}\label{table:distance_deviation} 
\end{table}

\subsection*{Implementation} All the experiments reported here were performed using R and the exact maximum clique detection algorithm proposed in~\cite{konc}.\footnote{We adapted the C++ implementation of this algorithm, which is available from \url{http://www.sicmm.org/~konc/maxclique/}, for our purposes.} All the accompanying visualisations were created in R.

\subsection*{Evaluation of the attack}
The matches and non-matches between the target and identification file gathered by the proposed graph theoretical linkage attack were classified as true positives (successful deanonymisation), false positives (failed deanonymisation), false negatives (records belonging to the same entity have been missed) and true negatives (records have been correctly classified as belonging to distinct entities). The quality measures considered are based on the number of true positives (\TP), false positives (\FP) and false negatives (\FN). More precisely, we consider
\begin{align*}
\pre&=\frac{\TP}{\TP+\FP} \quad \textit{(precision)}\text{,} \quad \text{ and }\\
\rec&=\frac{\TP}{\TP+\FN} \quad \textit{(recall)}\text{,}
\end{align*}
which are two standard measures in the evaluation of data linkage processes~\cite{christen_goiser}.
\enlargethispage{1cm}

\subsection*{Simulation design and results}

In our experiments, we varied the noise parameter $\sigma$ as well as the threshold parameter $\alpha$. For each parameter setup, the simulation was repeated $n=100$ times. The mean of precision and recall over all iterations for the chosen parameter setups can be found tables~\ref{table:results_prec_1} and~\ref{table:results_rec_1}. Visualisations of these results can be found in figures~\ref{fig:results} and~\ref{fig:results_dep_alpha}. In addition, typical outcomes of the graph theoretic linkage attack are visualised in figure~\ref{fig:typical_results}.

\enlargethispage{1cm}

\begin{table}[]
  \centering\scriptsize
  \begin{tabular}{*{1}{l}*{10}{r}}
    \toprule
    \( \alpha \) \textbar\ \( \sigma \) & \multicolumn{1}{c}{0.005} & \multicolumn{1}{c}{0.010} & \multicolumn{1}{c}{0.015} & \multicolumn{1}{c}{0.020} & \multicolumn{1}{c}{0.025} & \multicolumn{1}{c}{0.030} & \multicolumn{1}{c}{0.035} & \multicolumn{1}{c}{0.040} & \multicolumn{1}{c}{0.045} & \multicolumn{1}{c}{0.050} \\
    \midrule
    0.1 & 0.7800 & 0.7215 & 0.3415 & 0.2752 & 0.3232 & 0.2193 & 0.1549 & 0.1571 & 0.1136 & 0.1584 \\
    0.2 & 0.9717 & 0.8929 & 0.8769 & 0.8229 & 0.7569 & 0.6121 & 0.5657 & 0.5094 & 0.4780 & 0.4785 \\
    0.3 & 0.9831 & 0.9513 & 0.9047 & 0.8502 & 0.8255 & 0.7728 & 0.6862 & 0.6567 & 0.6073 & 0.5975 \\
    0.4 & 0.9829 & 0.9558 & 0.9037 & 0.8700 & 0.8358 & 0.7766 & 0.7411 & 0.6651 & 0.6428 & 0.6410 \\
    0.5 & 0.9808 & 0.9458 & 0.9133 & 0.8721 & 0.8374 & 0.7834 & 0.7505 & 0.6832 & 0.6526 & 0.6274 \\
    0.6 & 0.9830 & 0.9436 & 0.9102 & 0.8725 & 0.8315 & 0.7780 & 0.7430 & 0.6974 & 0.6604 & 0.6229 \\
    0.7 & 0.9803 & 0.9405 & 0.9087 & 0.8675 & 0.8255 & 0.7707 & 0.7434 & 0.6948 & 0.6556 & 0.6086 \\
    0.8 & 0.9795 & 0.9373 & 0.9008 & 0.8539 & 0.8027 & 0.7666 & 0.7248 & 0.6894 & 0.6484 & 0.6065 \\
    0.9 & 0.9764 & 0.9304 & 0.8884 & 0.8351 & 0.7954 & 0.7513 & 0.7017 & 0.6605 & 0.6159 & 0.5876 \\
    \bottomrule
  \end{tabular}
  \caption{\textbf{Average precision} in dependence on the parameters $\sigma$ and $\alpha$ over $n=100$ repetitions.}
  \label{table:results_prec_1}
\end{table}

\begin{table}[]
  \centering\scriptsize
  \begin{tabular}{*{1}{l}*{10}{r}}
    \toprule
    \( \alpha \) \textbar\ \( \sigma \) & \multicolumn{1}{c}{0.005} & \multicolumn{1}{c}{0.010} & \multicolumn{1}{c}{0.015} & \multicolumn{1}{c}{0.020} & \multicolumn{1}{c}{0.025} & \multicolumn{1}{c}{0.030} & \multicolumn{1}{c}{0.035} & \multicolumn{1}{c}{0.040} & \multicolumn{1}{c}{0.045} & \multicolumn{1}{c}{0.050} \\
    \midrule
    0.1 & 0.0664 & 0.0612 & 0.0302 & 0.0284 & 0.0326 & 0.0228 & 0.0168 & 0.0174 & 0.0130 & 0.0196 \\
    0.2 & 0.1166 & 0.1034 & 0.1126 & 0.1144 & 0.1120 & 0.0870 & 0.0864 & 0.0770 & 0.0714 & 0.0818 \\
    0.3 & 0.1586 & 0.1552 & 0.1556 & 0.1592 & 0.1642 & 0.1416 & 0.1398 & 0.1354 & 0.1292 & 0.1390 \\
    0.4 & 0.2084 & 0.2204 & 0.2112 & 0.2164 & 0.2162 & 0.1934 & 0.1966 & 0.1894 & 0.1828 & 0.2026 \\
    0.5 & 0.2774 & 0.2874 & 0.2754 & 0.2880 & 0.2722 & 0.2628 & 0.2526 & 0.2518 & 0.2404 & 0.2564 \\
    0.6 & 0.3622 & 0.3714 & 0.3308 & 0.3582 & 0.3364 & 0.3322 & 0.3166 & 0.3146 & 0.3084 & 0.3132 \\
    0.7 & 0.4402 & 0.4490 & 0.4250 & 0.4408 & 0.4198 & 0.4080 & 0.3976 & 0.3992 & 0.3754 & 0.3936 \\
    0.8 & 0.5638 & 0.5538 & 0.5420 & 0.5408 & 0.5096 & 0.5110 & 0.5204 & 0.5068 & 0.5104 & 0.4966 \\
    0.9 & 0.7202 & 0.7146 & 0.6960 & 0.6790 & 0.6568 & 0.6568 & 0.6838 & 0.6454 & 0.6658 & 0.6468 \\
    \bottomrule
  \end{tabular}
  \caption{\textbf{Average recall} in dependence on the parameters $\sigma$ and $\alpha$ over $n=100$ repetitions.}
  \label{table:results_rec_1}
\end{table}

\begin{figure}[]
\begin{subfigure}{.5\textwidth}
  \includegraphics[width=1.02\textwidth]{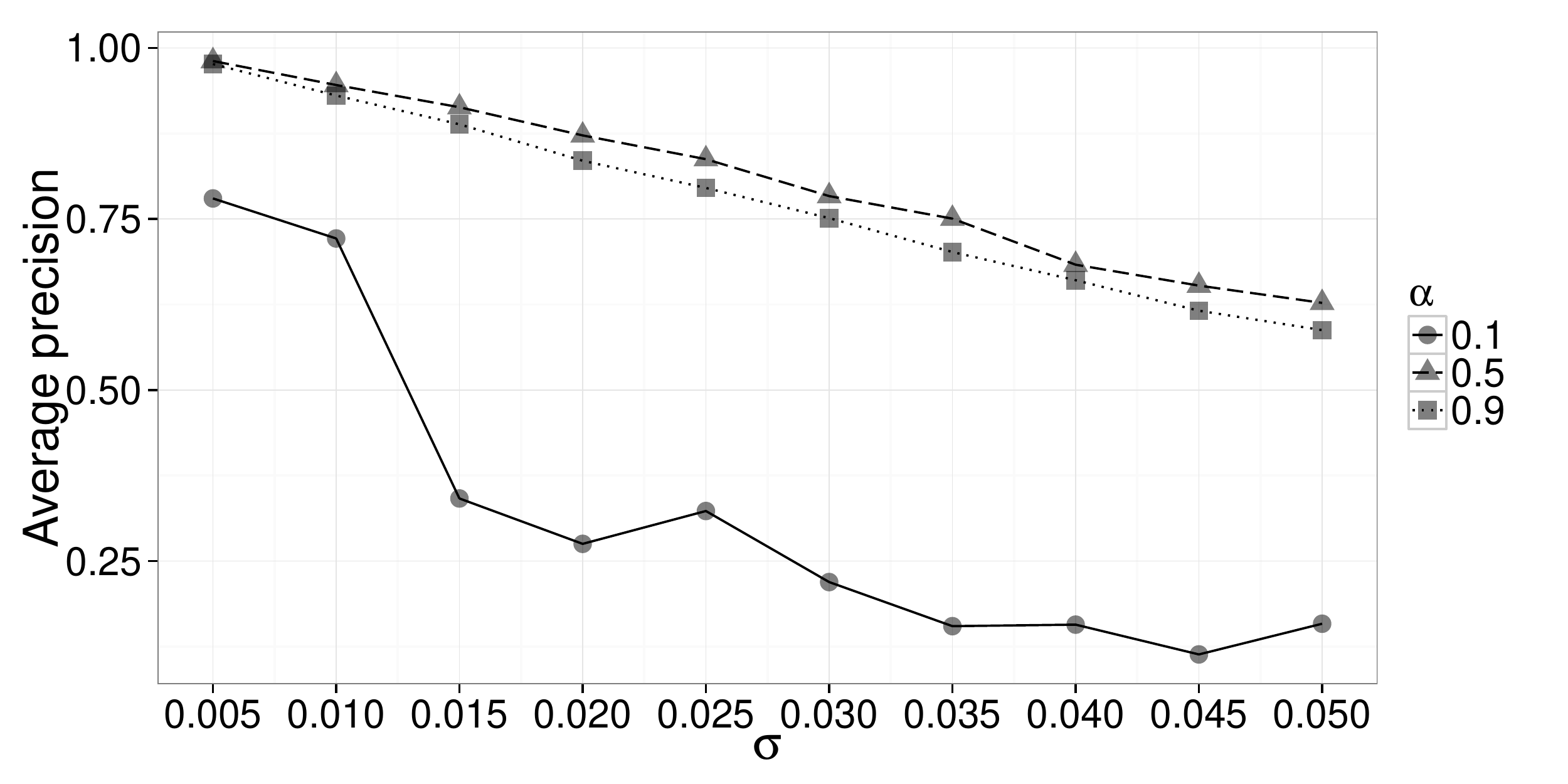}
  \caption*{}
\end{subfigure}
\begin{subfigure}{.5\textwidth}
  \includegraphics[width=1.02\textwidth]{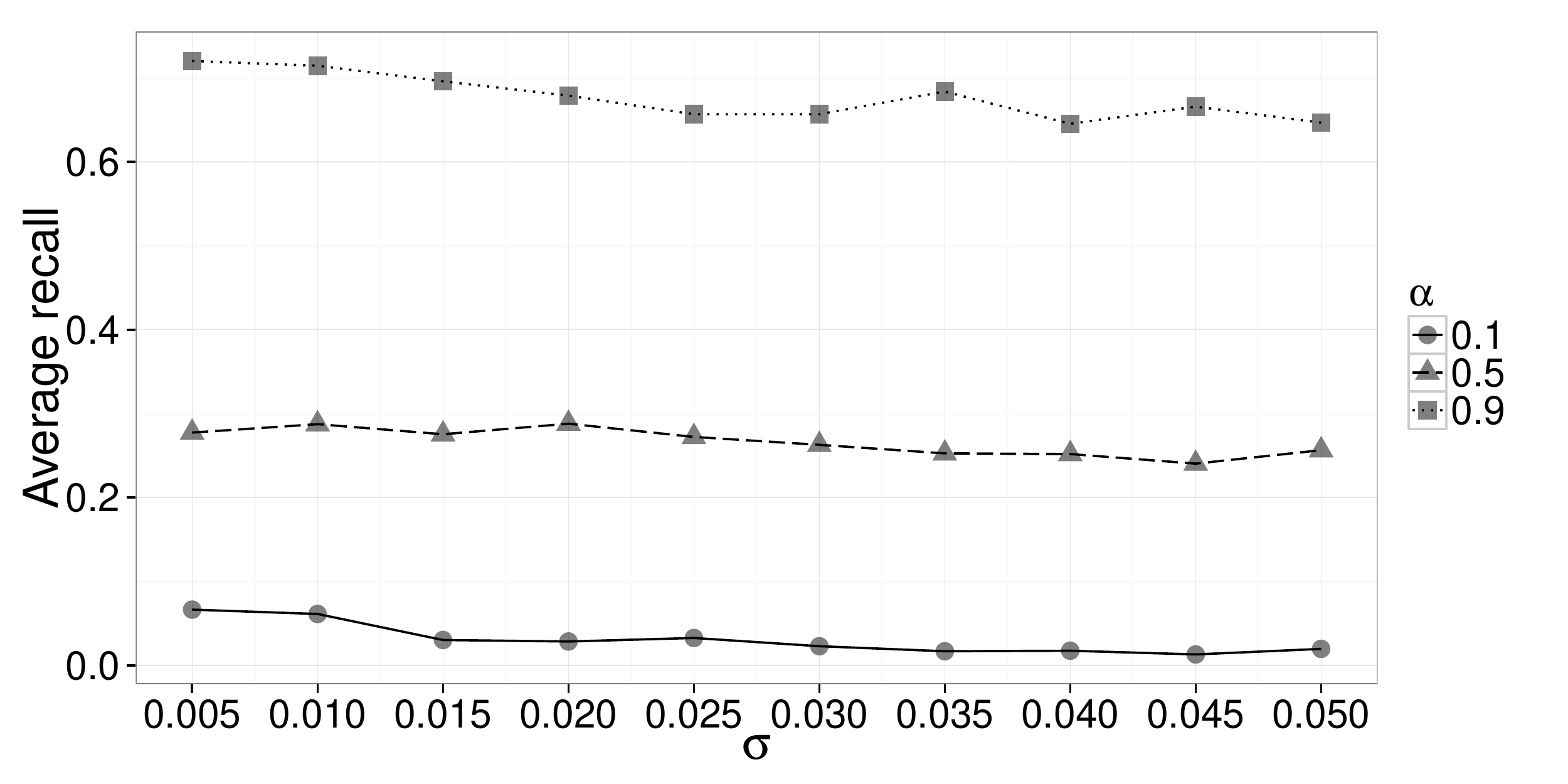}
  \caption*{}
\end{subfigure}%
\caption{Dependence of average precision (left) and recall (right) on the standard deviation $\sigma$ for different values of $\alpha$ (see tables~\ref{table:results_prec_1} and~\ref{table:results_rec_1}).}
\label{fig:results}
\end{figure}

\begin{figure}[]
\begin{subfigure}{.5\textwidth}
  \includegraphics[width=1.02\textwidth]{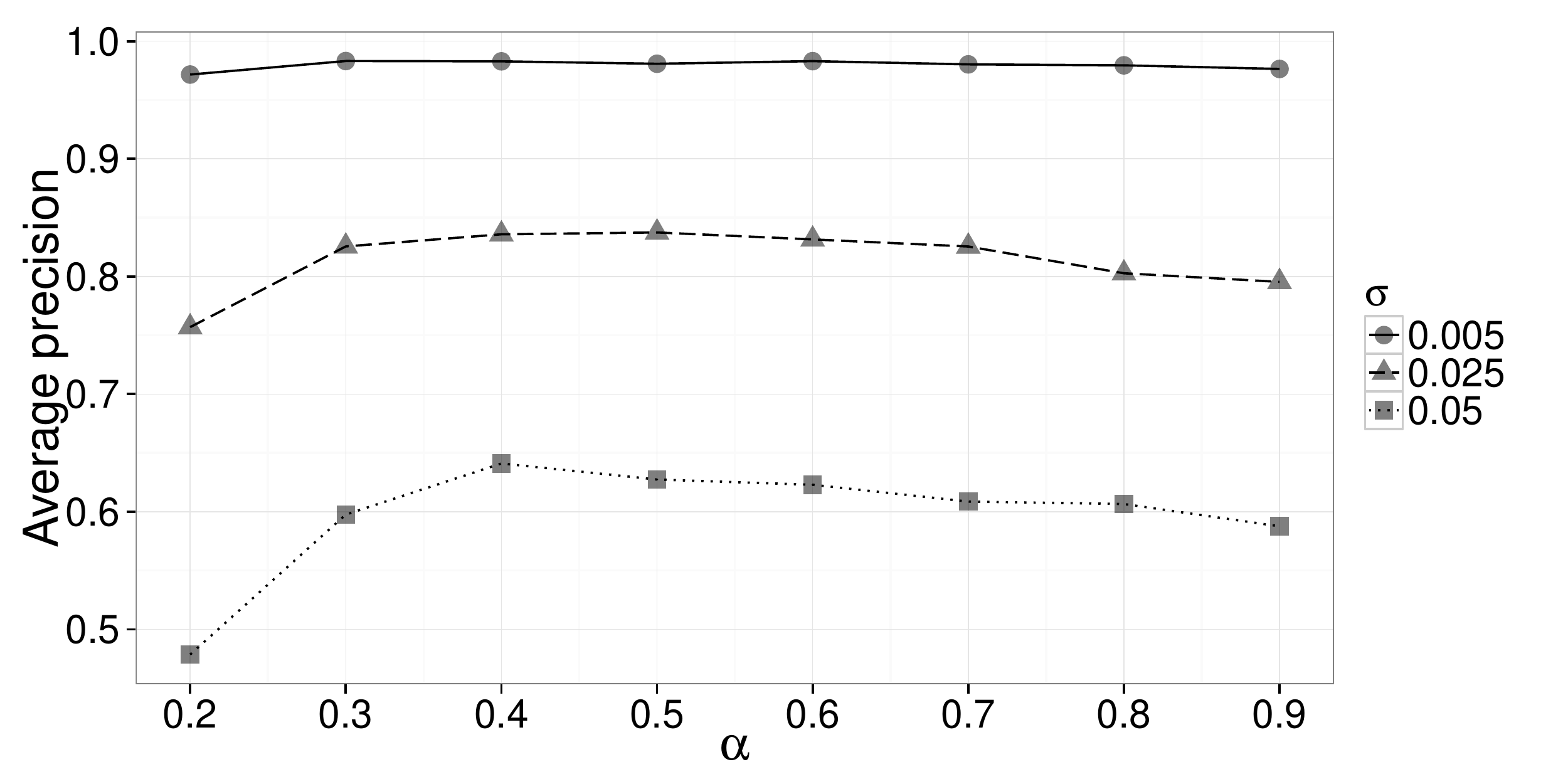}
  \caption*{}
\end{subfigure}
\begin{subfigure}{.5\textwidth}
  \includegraphics[width=1.02\textwidth]{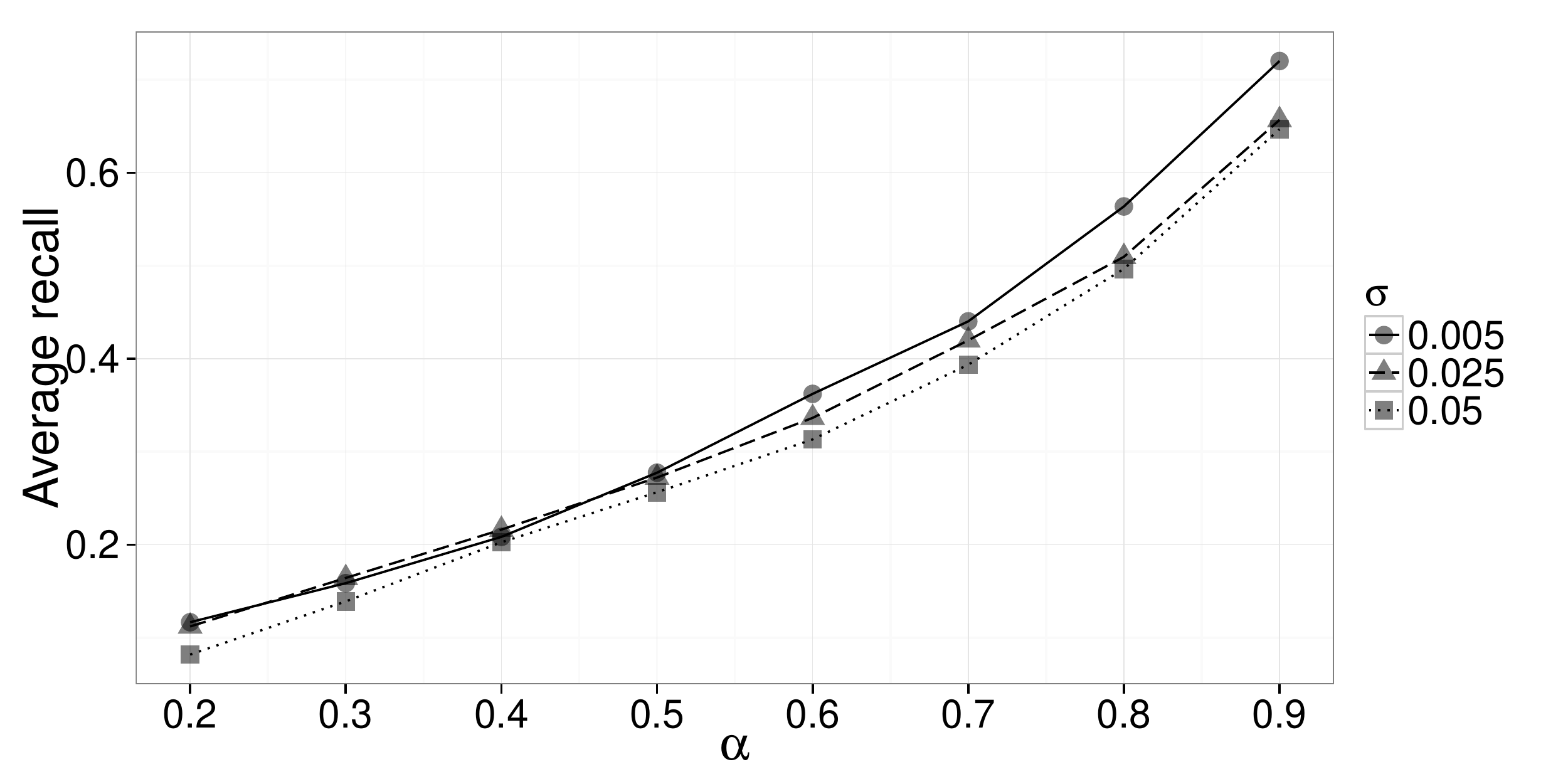}
  \caption*{}
\end{subfigure}%
\caption{Dependence of average precision (left) and recall (right) on the threshold parameter $\alpha$ for different values of $\sigma$ (see tables~\ref{table:results_prec_1} and~\ref{table:results_rec_1}).}
\label{fig:results_dep_alpha}
\end{figure}

\begin{figure}[]
\begin{subfigure}{.5\textwidth}
  \includegraphics[width=1.02\textwidth]{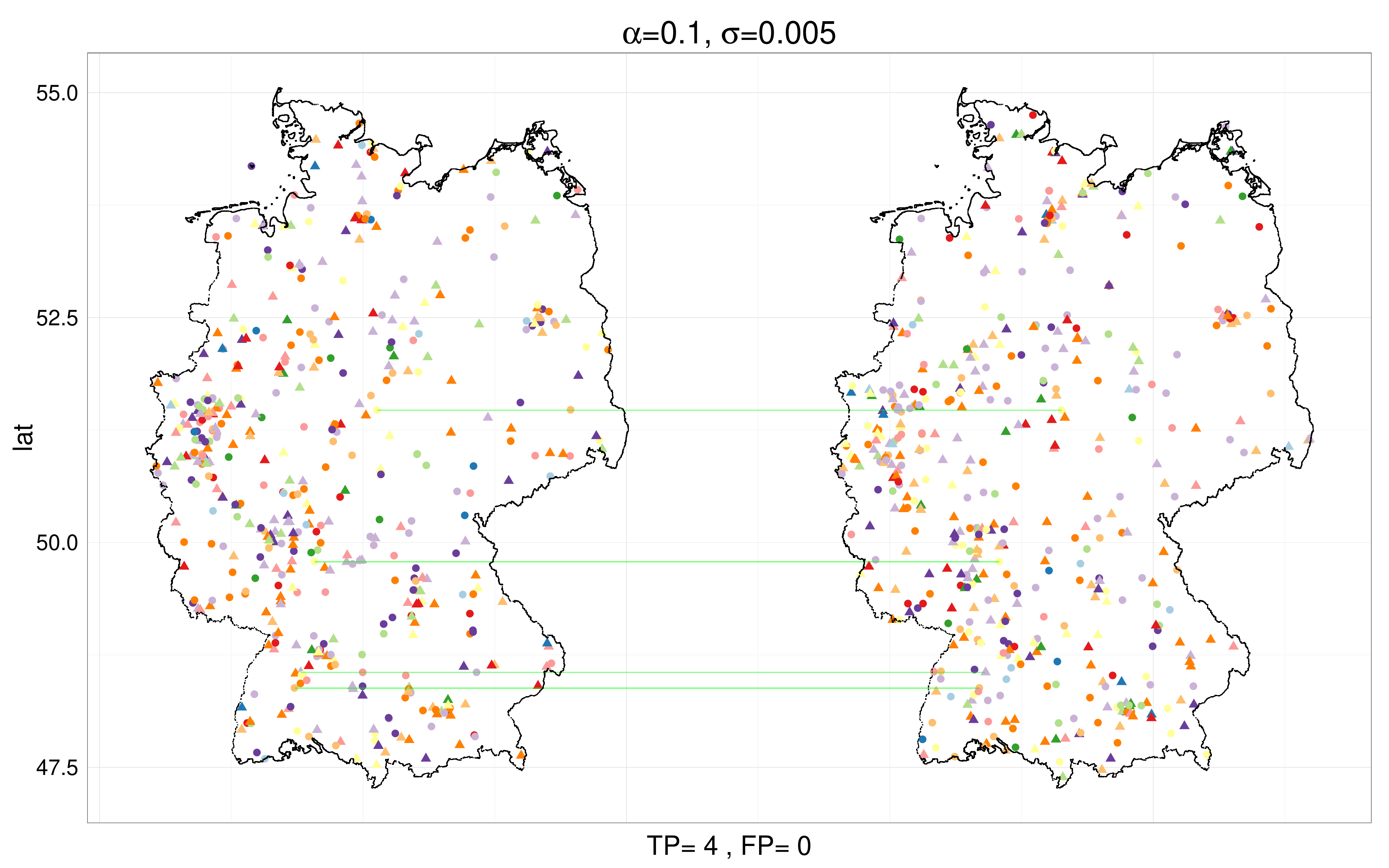}
  \caption*{}
\end{subfigure}
\begin{subfigure}{.5\textwidth}
  \includegraphics[width=1.02\textwidth]{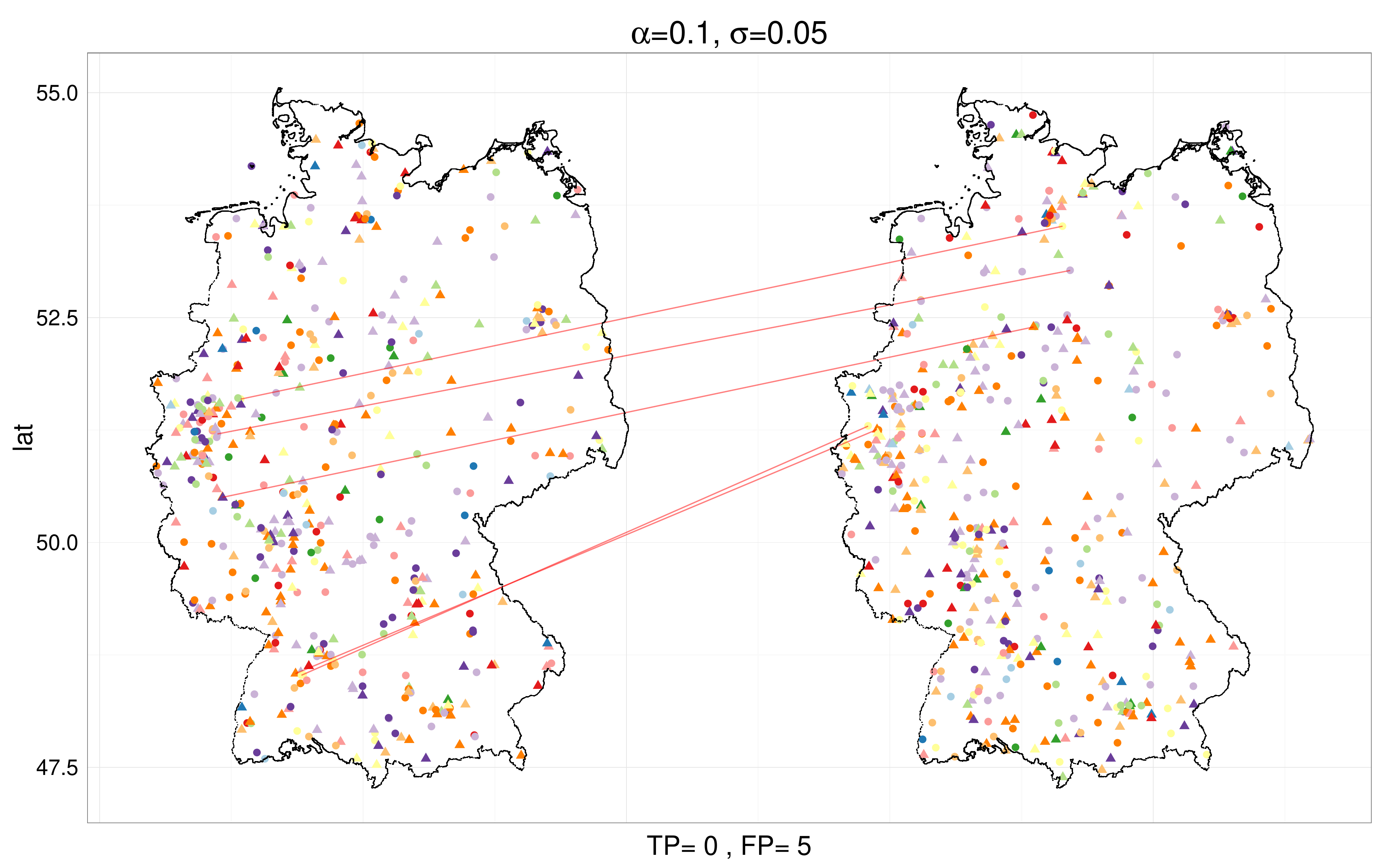}
  \caption*{}
\end{subfigure}%

\begin{subfigure}{.5\textwidth}
  \includegraphics[width=1.02\textwidth]{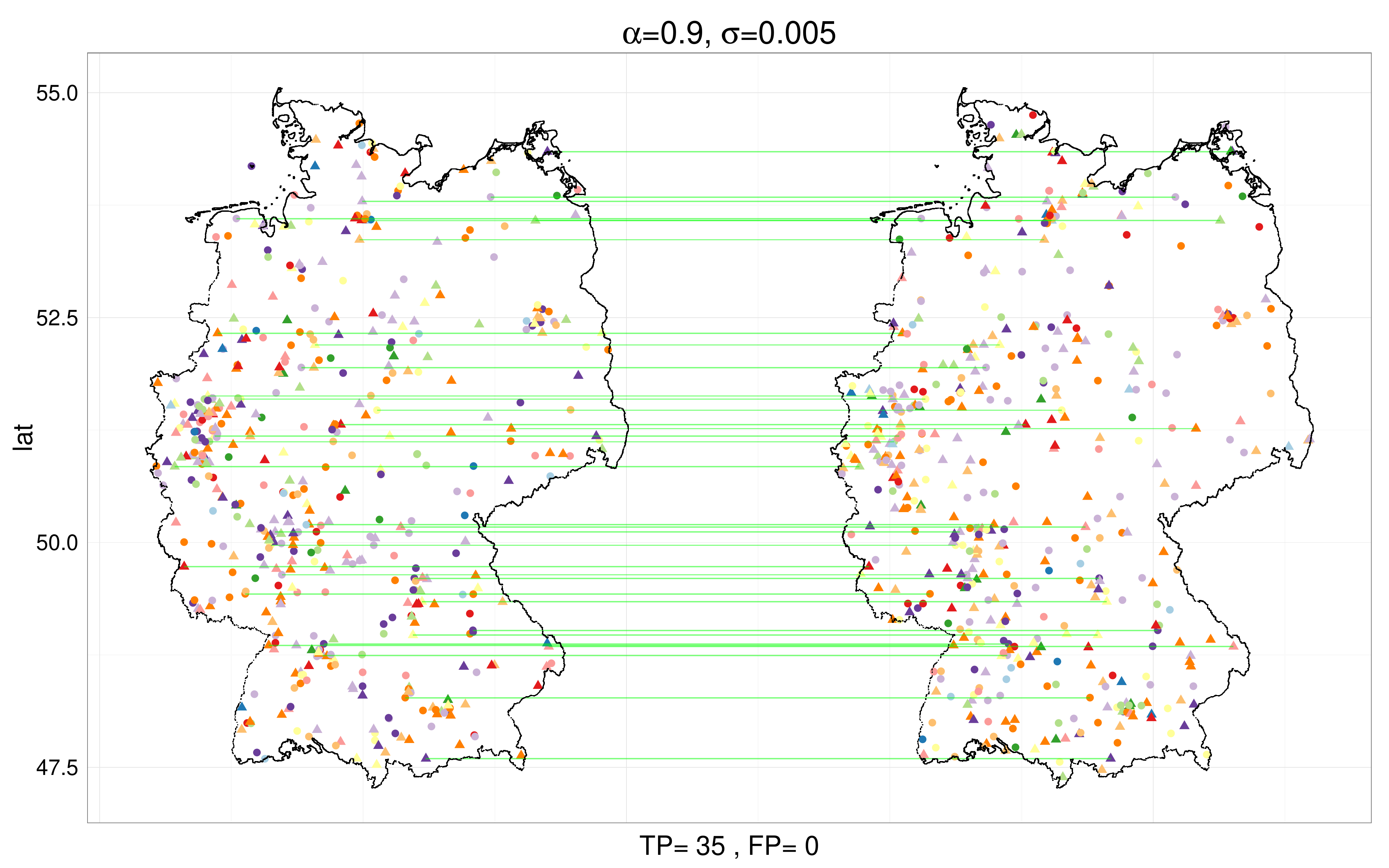}
  \caption*{}
\end{subfigure}
\begin{subfigure}{.5\textwidth}
  \includegraphics[width=1.02\textwidth]{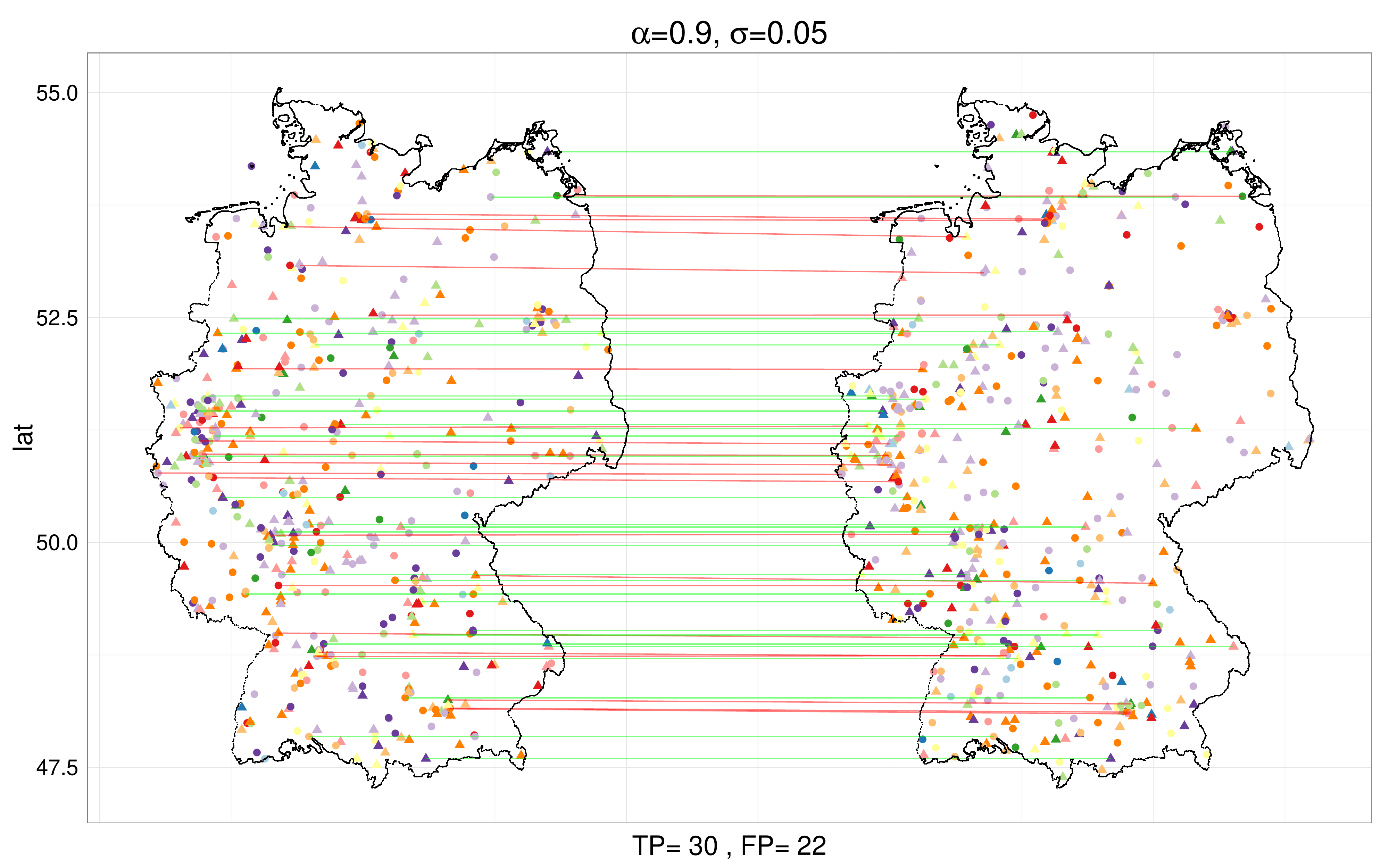}
  \caption*{}
\end{subfigure}%
\caption{Typical results of the graph theoretic linkage attack for different combinations of the noise parameter $\sigma$ (chosen by the data holder of the target file) and the threshold parameter $\alpha$ (chosen by the data snooper). Line segments between the target and identification file indicate matches made by the data snooper (the green lines indicate true, the red ones false positives). Larger values of $\alpha$ lead to more matches and a increase in recall, larger values of $\sigma$ to a decrease in precision.}
\label{fig:typical_results}
\end{figure}

\subsection*{Discussion}
The main effect of the threshold parameter $\alpha$ concerns the recall. The probability of detecting a common edge of the target and identification graph is approximately equal to $\alpha$. For this reason, higher values of $\alpha$ lead to a higher recall (see table~\ref{table:results_rec_1} and figures~\ref{fig:results} and~\ref{fig:results_dep_alpha}).

\smallskip

Simultaneously, the effect of $\alpha$ on the precision appears to be twofold:
On the one hand, for increasing $\alpha$ a larger portion of the overlap between the identification and target file can be successfully detected by the snooper, which makes false positives less likely (leading to a larger precision).
On the other hand, for too high values of $\alpha$ also the chance for non-common edges of the target and identification graph (but which coincide with respect to the vertex labels of their endpoints) to be classified as common edges increases leading to a slight decrease in precision. The latter phenomenon, together with the increase in recall for increasing $\alpha$ mentioned above, would reflect a trade-off between precision and recall, which is a well-known phenomenon in data linkage~\cite{christen_goiser}.
Thus, combining these two thoughts, for increasing $\alpha$ the precision should rapidly increase initially and then slightly decrease when $\alpha$ becomes too large. This expectation is confirmed by our experiments (see table~\ref{table:results_prec_1} and figures~\ref{fig:results} and~\ref{fig:results_dep_alpha}), although the decrease in precision when $\alpha$ becomes too large is not significant for the considered values of $\sigma$.

\smallskip

From the definition of $\approx$ (see the paragraph \textit{Fine-tuning of the attack} above), it is supposed that the recall does not change significantly in dependence on $\sigma$ because the probability of correctly detecting an edge should be nearly $\alpha$ (which is independent of $\sigma$). This non-dependence is impressively confirmed by the performed simulations and illustrated in figures~\ref{fig:results} and \ref{fig:results_dep_alpha}. However, $\sigma$ strongly influences the precision (for larger values of $\sigma$ the precision evidently decreases): The data snooper has to accept false positives (resulting in less precision) if she/he wants to achieve a certain predetermined recall.

\smallskip

Altogether, the simulations show that, in principle, a sufficient level of anonymity can be achieved by the addition of random noise to the input coordinates before computing the distance matrix.
However, this anonymity is not free, which is illustrated by means of the risk-utility (R-U) confidentiality map in figure~\ref{fig:RU}, where the risk of identity disclosure (measured by the average precision of the linkage attack; see also the discussion below) is plotted against the utility (measured as the reciprocal of the sample variance of the distance deviation $d-d'$ for the current value of $\sigma$ as recorded in table~\ref{table:distance_deviation}).
In the datasets considered in the simulation study $\sigma$ would have to be chosen large enough to guarantee at least some degree of anonymity, that useful analyses based on the distances would become difficult. For this reason, the development of distance modification techniques that guarantee a certain degree of anonymity, and make it possible to also conduct useful analyses on the anonymised data, will be an important aspect of future research.

\begin{figure}[]
\centering
\includegraphics[width=0.99\textwidth]{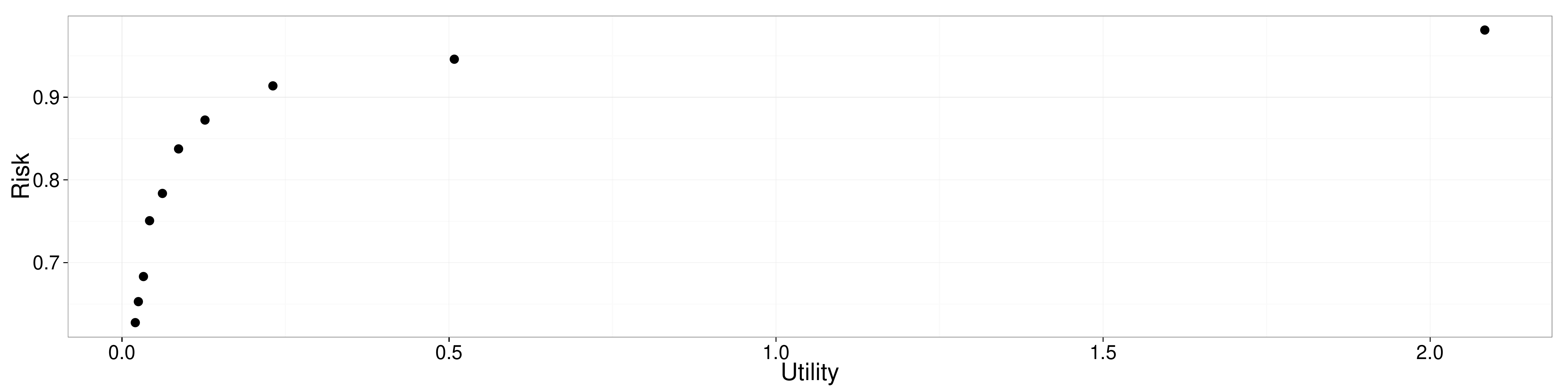}
\caption{R-U confidentiality map: The risk is measured by the average precision, whereas the utility is measured by the reciprocal of the sample variance in table~\ref{table:distance_deviation}. The threshold parameter $\alpha$ was chosen equal to $0.5$.}
\label{fig:RU}
\end{figure}

\smallskip

Note that in our specific example, the snooper would primarily attempt to achieve a high precision: In the case of geographic distances, a point is uniquely determined by the exact distances to three other points. If the snooper could deanonymise at least three entities successfully, exploiting this fact would be a good starting point to identify even more individuals. For arbitrary metric spaces, such a relationship does not hold in general, albeit the successful deanonymisation of some entities would also alleviate a snooper's work in this more general case.

\smallskip

Obviously, for distance modification techniques other than perturbation of the input coordinates, a snooper will have to modify the graph theoretical linkage attack, especially the definition of $\approx$. However, due to Kerckhoffs' principle, it has to be assumed that the snooper at least knows the distance modification technique used by the holder of the target file and exploits this knowledge in the precise construction of the attack. For instance, if noise is not added to the input coordinates before computing the distance matrix but rather to the distance matrix itself (a technique discussed in~\cite{kenthapadi}), the attack has to be slightly adapted. In this case, when defining the relation $\approx$ the quantiles of the noise distribution can be used directly, thereby making the empirical study on distance deviations originating from perturbation of the input coordinates unnecessary. Moreover, in this specific case it might be reasonable to further modify the attack by relaxing the (relatively strong) notion of a maximum clique to the less restrictive notion of a maximum quasi-clique, a relaxation which has been successfully applied in~\cite{fober} for the purpose of protein classification. In a similar way, our attack can be adapted to many other anonymisation techniques and thus provides a useful and flexible tool for the analysis of methods for distance-preserving anonymisation.

\section{Conclusion}\label{sec:conclusion}
In this article, we have introduced a novel graph theoretic linkage attack on microdata with additionally published (approximate) inter-record distances. In the special case of spatial distances, we have demonstrated -- on the basis of our test data -- that the release of distances increases the risk of identity disclosure unreasonably even if geographical coordinates have been perturbed by random Gaussian noise before the distances are calculated. Furthermore, we showed that augmenting the standard deviation of the added random noise will gradually lead to a sufficient level of anonymity, but also make the perturbed distances useless for further analysis.
Therefore, the development and analysis of anonymisation techniques for microdata in a metric space that guarantee a certain degree of anonymity but distort the distances as little as possible (particularly with regard to the applicability of data mining techniques) will be an important aspect of future research.

\section*{Acknowledgements}

This research has been supported by a grant from the German Research Foundation (DFG) to Rainer Schnell. I am grateful to Rainer Schnell for valuable suggestions that helped to improve the presentation of this paper.

\newpage

\appendix

\section{Example dataset: European poets}\label{sec:example_data}
\begin{table}[ht]
\centering
{\normalsize
\begin{tabular}{rlrll}
  \hline
 & name & yob & language & loc \\ 
  \hline
1 & Giovanni Boccaccio & 1313 & Italian & Firenze \\ 
  2 & Miguel de Cervantes & 1547 & Spanish & Alcala de Henares \\ 
  3 & Johann Wolfgang Goethe & 1749 & German & Frankfurt am Main \\ 
  4 & Moliere & 1622 & French & Paris \\ 
  5 & Dante Alighieri & 1265 & Italian & Firenze \\ 
  6 & Friedrich Schiller & 1759 & German & Marbach am Neckar \\ 
  7 & Jean-Baptiste Racine & 1637 & French & La Ferte-Milon \\ 
  8 & William Shakespeare & 1564 & English & Stratford-upon-Avon \\ 
  9 & Simone de Beauvoir & 1908 & French & Paris \\ 
  10 & Jean-Paul Sartre & 1905 & French & Paris \\ 
   \hline
\end{tabular}
}
\caption{Microdata containing information about famous European poets. The attribute \texttt{yob} contains the year of birth, and \texttt{loc} the birth location of the poets.}\label{table:poets}
\end{table}

\begin{table}[ht]
\centering
{\normalsize
\begin{tabular}{rrl}
  \hline
 & cob & language \\ 
  \hline
1 &  14 & Italian \\ 
  2 &  16 & Spanish \\ 
  3 &  18 & German \\ 
  4 &  17 & French \\ 
  5 &  13 & Italian \\ 
  6 &  18 & German \\ 
  7 &  17 & French \\ 
  8 &  16 & English \\ 
  9 &  20 & French \\ 
  10 &  20 & French \\ 
   \hline
\end{tabular}
}
\caption{The anonymised version of table \ref{table:poets} is obtained by removing the direct identifier \texttt{name}, generalising the year of birth (\texttt{yob}) to century of birth (\texttt{cob}) and removing the birth location (\texttt{loc}).}\label{table:poets_anonymized}
\end{table}

The distances between birth locations \texttt{loc} are stored in the distance matrix $D_1$:
$$D_1=\begin{pmatrix}{}
  0 & 1261 & 729 & 886 & 0 & 593 & 864 & 1341 & 886 & 886 \\ 
  1261 & 0 & 1424 & 1034 & 1261 & 1369 & 1093 & 1307 & 1034 & 1034 \\ 
  729 & 1424 & 0 & 479 & 729 & 137 & 414 & 762 & 479 & 479 \\ 
  886 & 1034 & 479 & 0 & 886 & 507 & 67 & 469 & 0 & 0 \\ 
  0 & 1261 & 729 & 886 & 0 & 593 & 864 & 1341 & 886 & 886 \\ 
  593 & 1369 & 137 & 507 & 593 & 0 & 449 & 856 & 507 & 507 \\ 
  864 & 1093 & 414 & 67 & 864 & 449 & 0 & 478 & 67 & 67 \\ 
  1341 & 1307 & 762 & 469 & 1341 & 856 & 478 & 0 & 469 & 469 \\ 
  886 & 1034 & 479 & 0 & 886 & 507 & 67 & 469 & 0 & 0 \\ 
  886 & 1034 & 479 & 0 & 886 & 507 & 67 & 469 & 0 & 0 \\ 
  \end{pmatrix}$$

\newpage

\begin{table}[h]
\centering
{\normalsize
\begin{tabular}{rlrll}
  \hline
 & name & cob & language & loc \\ 
  \hline
1 & Giovanni Boccaccio &  14 & Italian & Firenze \\ 
  2 & Miguel de Cervantes &  16 & Spanish & Alcala de Henares \\ 
  3 & Johann Wolfgang Goethe &  18 & German & Frankfurt am Main \\ 
  4 & Moliere &  17 & French & Paris \\ 
  5 & James Joyce &  19 & English & Dublin \\ 
  6 & Heinrich Heine &  18 & German & Duesseldorf \\ 
  7 & Pierre Corneille &  17 & French & Rouen \\ 
  8 & Publius Ovidius Naso &  -1 & Latin & Sulmona \\ 
  9 & Lope de Vega &  16 & Spanish & Madrid \\ 
  10 & August Strindberg &  19 & Swedish & Stockholm \\ 
   \hline
\end{tabular}
}
\caption{Identification microdata table used by the data snooper in example \ref{example:poets}.}\label{table:poets_intruder}
\end{table}

Geocoding of the locations from table~\ref{table:poets_intruder} using the R package \texttt{ggmap} and calculation of the mutual distances via the command \texttt{spDists} from the package \texttt{sp} yields the distance matrix $D_2$:

$$D_2=\begin{pmatrix}{}
  0 & 1260 & 731 & 887 & 1666 & 894 & 999 & 291 & 1290 & 1791 \\ 
  1260 & 0 & 1423 & 1033 & 1446 & 1427 & 1055 & 1457 & 30 & 2574 \\ 
  731 & 1423 & 0 & 479 & 1091 & 183 & 551 & 983 & 1447 & 1188 \\ 
  887 & 1033 & 479 & 0 & 782 & 412 & 112 & 1177 & 1052 & 1546 \\ 
  1666 & 1446 & 1091 & 782 & 0 & 919 & 671 & 1956 & 1450 & 1633 \\ 
  894 & 1427 & 183 & 412 & 919 & 0 & 450 & 1156 & 1448 & 1149 \\ 
  999 & 1055 & 551 & 112 & 671 & 450 & 0 & 1290 & 1071 & 1548 \\ 
  291 & 1457 & 983 & 1177 & 1956 & 1156 & 1290 & 0 & 1487 & 1942 \\ 
  1290 & 30 & 1447 & 1052 & 1450 & 1448 & 1071 & 1487 & 0 & 2595 \\ 
  1791 & 2574 & 1188 & 1546 & 1633 & 1149 & 1548 & 1942 & 2595 & 0 \\ 
  \end{pmatrix}$$

\end{document}